\documentclass[a4paper,11pt]{article}
\usepackage{amsfonts,amsmath}
\usepackage{graphicx}

\usepackage{amsthm}

\usepackage{amssymb}

\usepackage{url}

\usepackage[norelsize,lined,ruled,linesnumbered]{algorithm2e}

\newlength{\minitwocolumn}
\newtheorem{lemma}{Lemma}

\newtheorem{theorem}{Theorem}

\usepackage{color}
\begin{document}
\title{Max-Min and 1-Bounded Space Algorithms\\ for the Bin Packing Problem}
\author{
Hiroshi Fujiwara
\and
Rina Atsumi
\and
Hiroaki Yamamoto
}
\date{\today}
\maketitle              % typeset the header of the contribution
%%
%%
%%
%% \maketitle
%% 
%% 
\begin{abstract}
  %% In the bin packing problem, we are asked to pack all the given items
  In the (1-dimensional) bin packing problem, we are asked to pack all
  the given items
  into bins, each of capacity one, so that the number of non-empty
  bins is minimized.  Zhu~[Chaos, Solitons \& Fractals 2016] proposed
  an approximation algorithm $MM$ that sorts the item sequence in a
  non-increasing order by size at the beginning, and then repeatedly packs,
  into the current single open bin, first as many of the largest items
  in the remaining sequence as possible and then as many of the
  smallest items in the remaining sequence as possible.  In this paper
  we prove that the asymptotic approximation ratio of $MM$ is at most
  1.5.  Next, focusing on the fact that $MM$ is at the intersection of
  two algorithm classes, max-min algorithms and 1-bounded space
  algorithms, we comprehensively analyze the theoretical performance
  bounds of each subclass derived from the two classes.  Our results
  include a lower bound of 1.25 for the intersection of the two
  classes.  Furthermore, we extend the theoretical analysis over
  algorithm classes to the cardinality constrained bin packing problem.
  %%
  %% \keywords{bin packing problem \and online algorithm \and approximation
  %%   algorithms analysis}
\end{abstract}
\footnotetext{
  We are grateful to Toshiki Tsuchiya for useful discussions.  We
  would like to thank Hiromu Hashimoto for pointing out an error in
  the manuscript.
  This work was supported by JSPS KAKENHI Grant Numbers
  25K14990,
  20K11689,
  20K11676,
  23K11100, and
  20K11808.
}
\section{Introduction}
\label{sec:introduction}
The bin packing problem is a basic combinatorial optimization problem
in which, given a sequence of items, each of size at most one, the
number of non-empty bins is minimized under the constraint that each
item is contained in some bin of capacity one.  The bin packing
problem is known to be NP-complete~\cite{Garey:1979:CIG:578533}. Many
approximation algorithms that run in polynomial time and assign bins
as close to the minimum number as possible have been studied
(see~\cite{DBLP:conf/dagstuhl/1996oa,Korte:2012:COT:2190621} for
example).

%% Zhu proposed an algorithm $MM$ that sorts the item sequence in a
%% non-increasing order by size, and then repeatedly packs at the
%% beginning, into the current single open bin, first as many of the
%% largest items in the remaining item sequence as possible and then as
%% many of the smallest items in the remaining item sequence as
%% possible~\cite{ZHU201683}.
Zhu proposed an algorithm $MM$ that sorts the item sequence in a
non-increasing order by size at the beginning, and then repeatedly
packs into the current single open bin, first as many of the largest
items in the remaining item sequence as possible and then as many of
the smallest items in the remaining item sequence as
possible~\cite{ZHU201683}
(see
Algorithm~\ref{algo:mm} for the pseudocode)~\cite{ZHU201683}.
His paper shows through several computational experiments that
algorithm $MM$ has good practical performance.

In this paper we conduct a theoretical performance analysis of
algorithm $MM$ as well as the relevant classes of algorithms (see
Figure~\ref{fig:maxmin_unit_classes}).  We refer to an algorithm as a
\emph{max-min} algorithm if it first sorts the given item sequence in
a non-increasing order by size and then repeatedly packs either the head or
tail item in the currently remaining item sequence into a bin, without
seeing other items.
From the above explanation, $MM$ is obviously a max-min
algorithm.

There is a closely related class of algorithms.  We refer to an
algorithm as a \emph{pre-sorted online} algorithm if it first sorts
the given item sequence in a non-increasing order by size and then
repeatedly packs the head item in the currently remaining item
sequence into a bin, without seeing other items.  Algorithm $NFD$
(Next Fit Decreasing)~\cite{10.1137/0602019} and algorithm $FFD$
(First Fit Decreasing)~\cite{johnson1973near} are typical examples and
have been studied for a long time.  In the context of priority
algorithms~\cite{DBLP:journals/algorithmica/BorodinNR03,DBLP:journals/jco/YeB08},
a pre-sorted online algorithm can be viewed as a \emph{fixed priority}
algorithm with a non-increasing order by size.  Any pre-sorted online
algorithm is also a max-min algorithm, in the sense that it is allowed
to choose the head or tail item but only the head item is chosen.

Our research is further motivated by the potential of max-min
algorithms as heuristics for high-dimensional
packing problems to which pre-sorted online algorithms
are often applied
(see~\cite{CHRISTENSEN201763,MOMMESSIN2025106860} for example).
How useful will it be to pack even the smallest items?

Another aspect is the space complexity.  During the execution of an
algorithm, a bin is called \emph{open} if some item has already been
packed in it and the algorithm may pack further items into it.  The
algorithm keeps track of the sum of the sizes of the items in each
open bin.  An algorithm is called a \emph{bounded space} algorithm if
there exists a constant $B$ such that the algorithm maintains
simultaneously at most $B$ open bins (for example, algorithm
$Harmonic$~\cite{Lee:1985:SOB:3828.3833}).  We use the terminology a
\emph{$B$-bounded space} algorithm to specify the maximum number of
open bins.  An algorithm with no such constant $B$ is called an
\emph{unbounded space} algorithm (for example, algorithm $FF$ (First
Fit)~\cite{johnson1973near}).  Since algorithm $MM$ maintains a single
open bin, $MM$ is a 1-bounded space algorithm.

The performance of approximation algorithms is often measured by the
{\em asymptotic approximation
  ratio}~\cite{johnson1973near,DBLP:conf/dagstuhl/1996oa}.  Let
$ALG(I)$ denote the number of non-empty bins by an algorithm $ALG$ for
an item sequence $I$.  Let $OPT(I)$ denote the minimum number of bins
that are needed to contain all items in $I$.  The asymptotic
approximation ratio of $ALG$ is defined as
\begin{equation}
  R_{ALG} =
  \lim_{m \to \infty} \sup \Bigl\{ \frac{ALG(I)}{OPT(I)}
  \;\Big|\; OPT(I) \geq m \Bigr\}.
  \label{eq:none_acr}
\end{equation}
For the \emph{$k$-cardinality constrained} bin packing problem, which
is a variant with an additional constraint that at most $k$ items can
be packed in each bin, the asymptotic approximation ratio of $ALG$ is
defined as,
\begin{equation}
  R_{ALG} =
  \lim_{m \to \infty} \sup \Bigl\{ \frac{ALG(I)}{OPT_{k}(I)}
  \;\Big|\; OPT_{k}(I) \geq m \Bigr\},
  \label{eq:k_acr}
\end{equation}
where $OPT_{k}(I)$ denotes the minimum number of bins that are needed
to contain all items in $I$ so that at most $k$ items are packed in
each bin.  From now on, we will refer to the bin packing problem that
we have argued so far as the \emph{classic} bin packing problem to
distinguish the problem from the $k$-cardinality constrained bin
packing problem.
%%
%% \begin{figure}[t]
\begin{figure}
  \centering
  \includegraphics[width=0.9\linewidth]{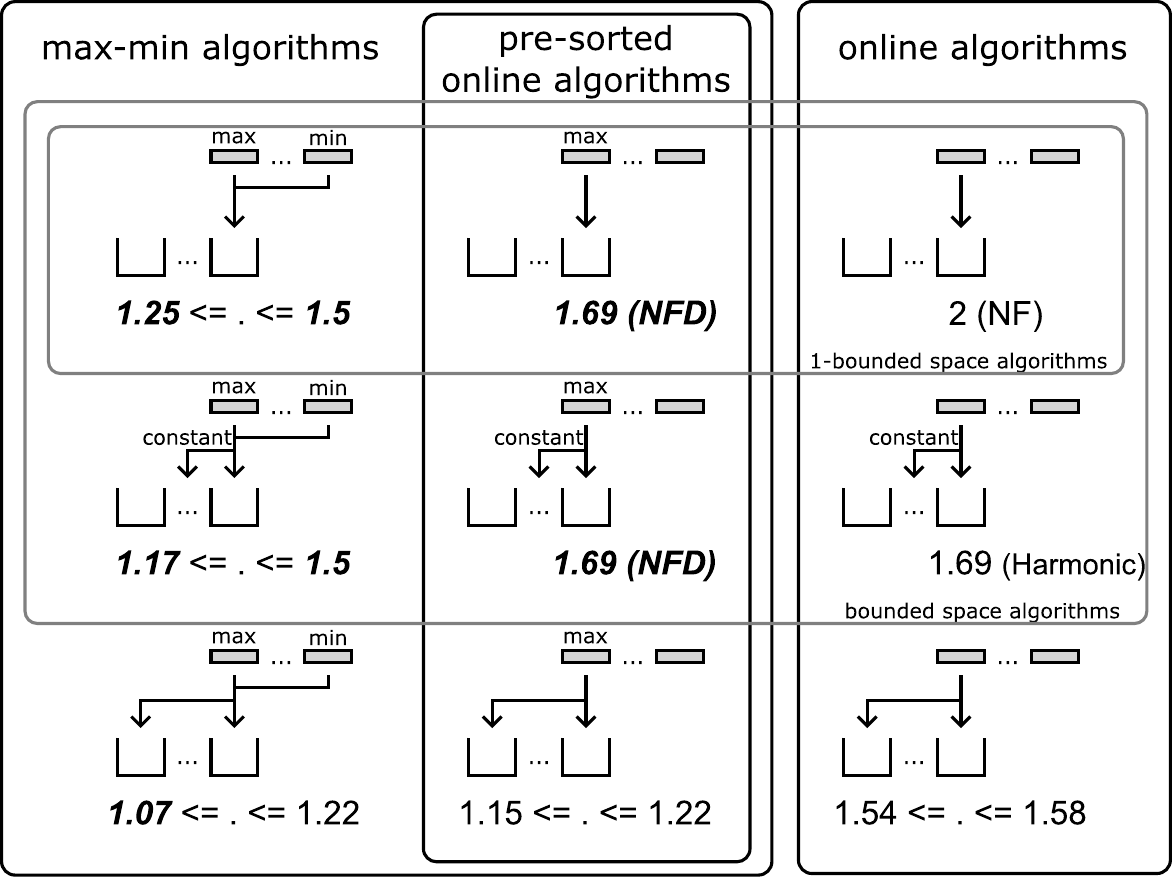}
  \caption{Illustration of algorithm classes and bounds on the
    best possible asymptotic approximation ratio
    for the classic bin packing problem.
    Bold italic font indicates bounds established in this paper.}
  \label{fig:maxmin_unit_classes}
\end{figure}
\begin{table*}
  \caption{Bounds on the best possible asymptotic approximation ratio of the
    classic bin packing problem.  Results with theorem references
    are established in this paper.  Parentheses mean that the result
    is a special case or corollary of another result.}
  \centering
  %% \Large
  \begin{tabular}{c|c|cc}
    type & space & 
    lower bound & upper bound \\
    \hline\hline
    max-min & 1-bounded &
    1.25 from Thm.~\ref{thm:none_maxmin_unit_lower} &
    1.5 from Thm.~\ref{thm:mm_cr} \\
    & bounded &
    1.17 from Thm.~\ref{thm:none_maxmin_bounded_lower} &
    (1.5 from Thm.~\ref{thm:mm_cr}) \\
    & unbounded &
    1.07 from Thm.~\ref{thm:none_maxmin_unbounded_lower} &
    $1.22$~\cite{johnson1973near} \\
    \hline
    pre-sorted online & 1-bounded &
    (1.69 from Thm.~\ref{thm:none_online_sorted_bounded_lower}) &
    $1.69$~\cite{10.1137/0602019} \\
    & bounded &
    1.69 from Thm.~\ref{thm:none_online_sorted_bounded_lower}
    & ($1.69$~\cite{10.1137/0602019}) \\
    & unbounded &
    $1.15$~\cite{DBLP:journals/tcs/BaloghBG12} &
    $1.22$~\cite{johnson1973near} \\
    \hline
    online & 1-bounded &
    $2$~\cite{johnson1973near} &
    $2$~\cite{johnson1973near} \\
    & bounded &
    $1.69$~\cite{Lee:1985:SOB:3828.3833} &
    $1.69$~\cite{Lee:1985:SOB:3828.3833} \\
    & unbounded &
    $1.54$~\cite{10.1007/s00453-021-00818-7} &
    $1.58$~\cite{conf/esa/BaloghBDEL18}
  \end{tabular}
  \label{tab:cr_without_cardi}
\end{table*}
\begin{table*}
  \caption{Bounds on the best possible asymptotic approximation ratio of the
    $k$-cardinality constrained bin packing problem for $k = 2$ and
    $3$, and for general $k$.  Results with theorem references are
    established in this paper.  Parentheses mean that the result is a
    special case or corollary of another result.  }
  \centering
  %% \small
  \begin{tabular}{c|c|c|cc}
    $k$ & type & space & 
    lower bound & upper bound \\
    \hline\hline
    2
    & max-min & 1-bounded &
    1 & 1 from Thm.~\ref{thm:tsuchiya} \\
    & & bounded &
    1 & (1 from Thm.~\ref{thm:tsuchiya}) \\
    & & unbounded &
    1 & (1~\cite{Tinhofer1995BinpackingAM}) \\
    \cline{2-5}
    & pre-sorted online & 1-bounded &
    (1.5~\cite{DBLP:journals/siamdm/Epstein06}) &
    (1.5 from Thm.~\ref{thm:k-item_online_sorted_unit_nextfit_k}) \\
    & & bounded &
    (1.5~\cite{DBLP:journals/siamdm/Epstein06}) &
    (1.5~\cite{DBLP:journals/siamdm/Epstein06}) \\
    & & unbounded &
    1 & 1~\cite{Tinhofer1995BinpackingAM} \\
    \cline{2-5}
    & online & 1-bounded &
    (2 from Thm.~\ref{thm:k-item_online_general_unit_lower}) &
    2~(trivial) \\
    & & bounded &
    (1.5~\cite{DBLP:journals/siamdm/Epstein06}) &
    (1.5~\cite{DBLP:journals/siamdm/Epstein06}) \\
    & & unbounded &
    1.43~\cite{BALOGH202034} & 1.45~\cite{DBLP:journals/dam/BabelCKK04} \\
    \hline\hline
    %%%%%%%%%%%%%%%%%%%%
    3
    & max-min & 1-bounded &
    1.33 from Thm.~\ref{thm:3-item_maxmin_unit_lower} &
    (1.5 from Thm.~\ref{thm:mm_k_cr}) \\
    & & bounded &
    (1.07 from Thm.~\ref{thm:none_maxmin_unbounded_lower}) &
    (1.5 from Thm.~\ref{thm:mm_k_cr}) \\
    & & unbounded &
    (1.07 from Thm.~\ref{thm:none_maxmin_unbounded_lower}) &
    (1.5 from Thm.~\ref{thm:mm_k_cr}) \\
    \cline{2-5}
    & pre-sorted online & 1-bounded & 
    (1.83~\cite{DBLP:journals/siamdm/Epstein06}) &
    (1.83 from Thm.~\ref{thm:k-item_online_sorted_unit_nextfit_k}) \\
    & & bounded &
    (1.83~\cite{DBLP:journals/siamdm/Epstein06}) &
    (1.83~\cite{DBLP:journals/siamdm/Epstein06}) \\
    & & unbounded &
    1.14~\cite{DBLP:conf/dagstuhl/1996oa} &
    (1.75~\cite{DBLP:journals/siamdm/Epstein06}) \\
    \cline{2-5}
    & online & 1-bounded &
    (2.33 from Thm.~\ref{thm:k-item_online_general_unit_lower}) &
    (2.33 from Thm.~\ref{thm:k-item_online_general_unit_nextfit_k}) \\
    & & bounded &
    (1.83~\cite{DBLP:journals/siamdm/Epstein06}) &
    (1.83~\cite{DBLP:journals/siamdm/Epstein06}) \\
    & & unbounded &
    1.56~\cite{BALOGH202034} & 1.75~\cite{DBLP:journals/siamdm/Epstein06} \\
    \hline\hline
    %%%%%%%%%%%%%%%%%%%%
    $k$
    & max-min & 1-bounded &
    (1.5 for $k \geq 6$
      from Thm.~\ref{thm:k-item_maxmin_bounded_lower}) &
    $\lambda_{k} - \frac{1}{k} \to 2.69$
      from Thm.~\ref{thm:mm_k_cr} \\
    & & bounded & 
    1.5 for $k \geq 6$
      from Thm.~\ref{thm:k-item_maxmin_bounded_lower} &
    ($\lambda_{k} - \frac{1}{k} \to 2.69$
      from Thm.~\ref{thm:mm_k_cr}) \\
    & & unbounded &
    (1.07 from Thm.~\ref{thm:none_maxmin_unbounded_lower}) &
    ($\nu_{k} \to 2$~\cite{DBLP:journals/dam/BabelCKK04}) \\
    \cline{2-5}
    & pre-sorted online & 1-bounded &
    ($\lambda_{k} \to 2.69$~\cite{DBLP:journals/siamdm/Epstein06}) &
    $\lambda_{k} \to 2.69$
      from Thm.~\ref{thm:k-item_online_sorted_unit_nextfit_k} \\
    & & bounded &
    $\lambda_{k} \to 2.69$~\cite{DBLP:journals/siamdm/Epstein06} &
    $\lambda_{k} \to 2.69$~\cite{DBLP:journals/siamdm/Epstein06} \\
    & & unbounded &
    $1.15$ for $k \geq 5$~\cite{DBLP:journals/tcs/BaloghBG12} &
    ($\nu_{k} \to 2$~\cite{DBLP:journals/dam/BabelCKK04}) \\
    \cline{2-5}
    & online & 1-bounded &
    $3 - \frac{2}{k} \to 3$
      from Thm.~\ref{thm:k-item_online_general_unit_lower} &
    $3 - \frac{2}{k} \to 3$
      from Thm.~\ref{thm:k-item_online_general_unit_nextfit_k} \\
    & & bounded &
    $\lambda_{k} \to 2.69$~\cite{DBLP:journals/siamdm/Epstein06} &
    $\lambda_{k} \to 2.69$~\cite{DBLP:journals/siamdm/Epstein06} \\
    & & unbounded &
    $\mu_{k} \to 2$~\cite{BALOGH202034} &
    $\nu_{k} \to 2$~\cite{DBLP:journals/dam/BabelCKK04}
  \end{tabular}
  \label{tab:cr_cardi}
\end{table*}
%%
%%%%%%%%%%%%%%%%%%%%%%%%%%%%%%%%%%%%%%%%%%%%%%%%%%%%%%%%%%%%%%%%%%%%%%
\subsection{Our Contribution}
\label{subsec:ourcontribution}
It was a breakthrough that for the classic bin packing problem, Lee
and Lee clarified a theoretical performance gap between online bounded
space algorithms and online unbounded space
algorithms~\cite{Lee:1985:SOB:3828.3833}, where an \emph{online}
algorithm is an algorithm that repeatedly packs the head item in the currently
remaining item sequence into a bin, without seeing other items.  In
this paper, extending this perspective, we comprehensively analyze
theoretical performance bounds for algorithm classes (see
Table~\ref{fig:maxmin_unit_classes}).

We make some preliminary definitions to express the value of the
asymptotic competitive ratio.  Define the sequence $\{\pi_{i}\}_{i =
  1}^{\infty}$ by $\pi_{1} = 2$ and $\pi_{i+1} = \pi_{i}(\pi_{i} - 1)
+ 1$ for $i \geq 1$. Let $\gamma = \sum_{i=1}^{\infty}
\frac{1}{\pi_{i}-1}$, which is known to converge to approximately
$1.69$~\cite{10.1137/0602019}.  Define the sequence
$\{\lambda_{j}\}_{j = 1}^{\infty}$ by $\lambda_{j} = \sum_{i=1}^{j}
\max \{\frac{1}{\pi_{i}-1}, \frac{1}{j}\}$ for $j \geq 1$. It is known
that $\lim_{j \to \infty} \lambda_{j} = \gamma + 1 \approx
2.69$~\cite{DBLP:journals/siamdm/Epstein06}. Throughout this paper,
numerical rounding is always done to the nearest value.  The theorems
in Sections~\ref{sec:none} and~\ref{sec:k-item}, each of which states
an inequality on the number of bins, provide the following bounds on
the asymptotic approximation ratio:

(I)~We obtain the following results for the classic bin packing
problem (see Figure~\ref{fig:maxmin_unit_classes} and
Table~\ref{tab:cr_without_cardi}).  We prove $R_{MM} \leq \frac{3}{2}
= 1.5$ (Theorem~\ref{thm:mm_cr}).  Compare with $R_{NFD} = \gamma
\approx 1.69$~\cite{10.1137/0602019}.  $NFD$ can be seen as a
pre-sorted online 1-bounded space algorithm.  Our result explains the
extent to which performance can be improved by packing also the
smallest items in the remaining sequence.

On the other hand, we show $R_{ALG} \geq \frac{5}{4} = 1.25$ for any
max-min 1-bounded space algorithm $ALG$
(Theorem~\ref{thm:none_maxmin_unit_lower}).  Compare with $R_{FFD} =
\frac{11}{9} \approx 1.22$~\cite{johnson1973near}.  $FFD$ is a
pre-sorted online unbounded space algorithm.  Our result reveals that
with only 1-bounded space, any max-min algorithm cannot achieve better
performance than $FFD$. In addition, we provide a lower bound of
$\frac{7}{6} \approx 1.17$ for max-min bounded space algorithms
(Theorem~\ref{thm:none_maxmin_bounded_lower}), and that of
$\frac{16}{15} \approx 1.07$ on max-min unbounded space algorithms
(Theorem~\ref{thm:none_maxmin_unbounded_lower}).

Furthermore, we give a lower bound for pre-sorted online bounded space
algorithms of $\gamma \approx 1.69$
(Theorem~\ref{thm:none_online_sorted_bounded_lower}). Combining with
$R_{NFD} = \gamma$ introduced above~\cite{10.1137/0602019}, we
conclude that the best possible performance of pre-sorted online
algorithms is not affected by a single open bin available or a
constant number of open bins available. This result contrasts with the
gap between the asymptotic competitive ratios of the best possible
online 1-bounded space algorithm and the best possible online bounded
space algorithm, which are 2~\cite{johnson1973near} and $\gamma
\approx 1.69$~\cite{Lee:1985:SOB:3828.3833}.

(II)~Our results for the $k$-cardinality constrained bin packing
problem are as follows (see Table~\ref{tab:cr_cardi}).  For the
$k$-cardinality-constrained version of algorithm $MM$, denoted by
$MM_{k}$,
(see Algorithm~\ref{algo:mmk} for the pseudocode),
we show
$R_{MM_{k}} \leq \lambda_{k} - \frac{1}{k}$ for each $k \geq 2$
(Theorems~\ref{thm:tsuchiya} and~\ref{thm:mm_k_cr}), which breaks the
lower bound barrier of $\lambda_{k}$ for pre-sorted online bounded
space $k$-cardinality constrained
algorithms~\cite{DBLP:journals/siamdm/Epstein06}.  For example, the
values of $\lambda_{k} - \frac{1}{k}$ are: $\frac{3}{2} = 1.5$ for $k
= 3$, $\frac{7}{4} = 1.75$ for $k = 4$, and $\frac{19}{10} = 1.9$ for
$k = 5$.  As $k \to \infty$, we have $\lambda_{k} - \frac{1}{k} \to
\gamma + 1 \approx 2.69$.  In particular, for $k = 2$ and $3$,
$MM_{k}$ also breaks the lower bounds of $\frac{10}{7} \approx 1.43$
and $1.56$~\cite{BALOGH202034}, respectively, for online unbounded
space $k$-cardinality constrained algorithms.

We provide a lower bound for max-min 1-bounded space algorithms:
$\frac{4}{3} \approx 1.33$ for $k = 3$
(Theorem~\ref{thm:3-item_maxmin_unit_lower}), $\frac{4}{3} \approx
1.33$ for $k = 4$, $\frac{17}{12} \approx 1.42$ for $k = 5$, and
$\frac{3}{2} = 1.5$ for $k \geq 6$
(Theorem~\ref{thm:k-item_maxmin_bounded_lower}).
We show that the $k$-cardinality-constrained version of algorithm
$NFD$ is a best possible pre-sorted online bounded space algorithm
(Theorem~\ref{thm:k-item_online_sorted_unit_nextfit_k}), by combining
with the lower bound by Epstein~\cite{DBLP:journals/siamdm/Epstein06}.
We prove that the $k$-cardinality-constrained version of algorithm
$NF$ (Next Fit) is a best possible online 1-bounded space algorithm
(Theorems~\ref{thm:k-item_online_general_unit_nextfit_k}
and~\ref{thm:k-item_online_general_unit_lower}).
%%
%%%%%%%%%%%%%%%%%%%%%%%%%%%%%%%%%%%%%%%%%%%%%%%%%%%%%%%%%%%%%%%%%%%%%%
\subsection{Previous Results}
\label{subsec:previous_results}
(I)~In the literature we can find the following results for the
classic bin packing problem: While a tight bound of $\gamma \approx
1.69$ for online bounded space algorithms, achieved by algorithm
$Harmonic$, is established by Lee and
Lee~\cite{Lee:1985:SOB:3828.3833}, a tight bound for online unbounded
space algorithms is unknown so far.  The current best upper and lower
bounds are $1.58$~\cite{conf/esa/BaloghBDEL18} and $\frac{1363 -
  \sqrt{1387369}}{120} \approx
1.54$~\cite{10.1007/s00453-021-00818-7}, respectively.  For online
%% 1-bounded space algorithms, it is known that a tight bound of 2 
1-bounded space algorithms, it is known that a tight bound of 2 is
achieved by $NF$ (Next Fit)~\cite{johnson1973near}.

Pre-sorted online algorithms such as $NFD$~\cite{10.1137/0602019} and
$FFD$~\cite{johnson1973near} have been studied by many researchers for
a long time.  The refinement of the analysis of $FFD$ has attracted
great interest~\cite{Dsa2007TheTB,DOSA201313}. Balogh et al.\ gave a
lower bound of $\frac{54}{47}\approx 1.15$ for pre-sorted online
unbounded space algorithms~\cite{DBLP:journals/tcs/BaloghBG12}.

(II)~The results below are known for the $k$-cardinality constrained
bin packing problem.
Epstein
provided a tight bound of $\lambda_{k}$ for online bounded space
algorithms, achieved by algorithm $CCH_{k}$ (Cardinality Constrained
Harmonic)~\cite{DBLP:journals/siamdm/Epstein06}.  This is the
counterpart to the work of Lee and Lee mentioned
above~\cite{Lee:1985:SOB:3828.3833}.  Since the lower bound is proved
by a non-increasing item sequence, the tight bound also applies to
pre-sorted online bounded space algorithms.

Babel et al.\ designed an online unbounded space algorithm with an
asymptotic approximation ratio of $\nu_{k} = 2 + \frac{k +
  \eta_{k}(\eta_{k} - 3)} {(k - \eta_{k} + 1) \eta_{k}}$, where
$\eta_{k} = \lceil \frac{- k + \sqrt{k^{3} - 2k}}{k - 2} \rceil$ for
each $k$~\cite{DBLP:journals/dam/BabelCKK04}.  For $k = 2$ and $3$,
there are better algorithms: one with an asymptotic approximation
ratio of $1+1/\sqrt{5} \approx 1.45$ for $k =
2$~\cite{DBLP:journals/dam/BabelCKK04} and one with an asymptotic
approximation ratio of $\frac{7}{4} = 1.75$ for $k =
3$~\cite{DBLP:journals/siamdm/Epstein06}.  Balogh et al.\ also gave a
lower bound for online unbounded space algorithms~\cite{BALOGH202034}:
$\frac{10}{7} \approx 1.43$ for $k = 2$, $1.56$ for $k = 3$, and
$\mu_{k} = \frac{ 2k + k^{2} - k^{3} + \sqrt{ k^{6} - 2 k^{5} - 3
    k^{4} + 12 k^{3} - 12 k^{2} + 4 k}}{2}$ for general $k$.

The lower bound of $\frac{54}{47} \approx 1.15$ for pre-sorted online
unbounded space algorithms by Balogh et
al.~\cite{DBLP:journals/tcs/BaloghBG12} is valid for $k \geq 5$.  The
book~\cite{DBLP:conf/dagstuhl/1996oa} introduces a lower bound of
$\frac{8}{7} \approx 1.14$ for pre-sorted online unbounded space
algorithms by Csirik, Galambos, and Tur\'{a}n with their proof, which
is still the best each for $k = 3$ and $4$.  Tinhofer showed that the
$2$-cardinality-constrained version of $FFD$ achieves an asymptotic
approximation ratio of 1~\cite{Tinhofer1995BinpackingAM}.
%%
%% 202405200841 和文書き始め。
%%
%%%%%%%%%%%%%%%%%%%%%%%%%%%%%%%%%%%%%%%%%%%%%%%%%%%%%%%%%%%%%%%%%%%%%%
\section{The Classic Bin Packing Problem}
\label{sec:none}
%%
%% 202403111042 和文執筆開始。
%% 202503191438 英訳開始。
%%
We formulate the \emph{classic} bin packing problem.  (We have added
the term classic to distinguish the problem from the $k$-cardinality
constrained version.)  Let $\mathbb{N}$ be the set of all positive
integers.  The input is an item sequence $I = (a_{1}, a_{2}, \ldots,
a_{n}) \in (0, 1]^{n}$, where each item $i$ is of size $a_{i}$.  We
  refer to an item of size $a$ as an \emph{$a$-item}.  An item
  sequence $I = (a_{1}, a_{2}, \ldots, a_{n}) \in (0, 1]^{n}$ is said
    to be \emph{sorted} if $a_{1} \geq a_{2} \geq \cdots \geq a_{n}$.

A solution to the classic bin packing problem is an assignment $f:
\{1, 2, \ldots, n\} \to \mathbb{N}$ which maps each item index to a
bin index and satisfies the constraint that for all bin indices $j \in
\mathbb{N}$, $\sum_{i: f(i) = j} a_{i} \leq 1$ holds.  The constraint
means that the sum of the sizes of the items in each bin does not
exceed $1$.  If bin $j$ is assigned to item $i$, i.e., $f(i) = j$, we
say that bin $j$ \emph{contains} item $i$.  An item is said to be a
\emph{sole} item if a bin is assigned exclusively to the item.

Given an item sequence $I$ and an assignment $f$ for $I$, the number
of non-empty bins is the number of elements in the range of assignment
$f$, denoted by $|f(I)|$.  The goal of the classic bin packing
problem is to minimize $|f(I)|$.

An algorithm for the classic bin packing problem is to output an
assignment $f: \{1, 2, \ldots, n\} \to \mathbb{N}$.  If an algorithm
assigns bin $j$ to item $i$, we interchangeably say that the algorithm
\emph{packs} item $i$ into bin $j$.  Without loss of generality, we
assume that when an algorithm is about to assign one of the bins that
have not previously had any items assigned to it (often referred to as
\emph{opening} the bin), the algorithm always chooses the bin of the
\emph{smallest} index.  We refer to a bin with a smaller
(respectively, larger) index as an \emph{earlier} (respectively, a
\emph{later}) bin.

We write the number of non-empty bins of an assignment obtained by an
algorithm $ALG$ for an item sequence $I$ as $ALG(I)$, without
specifying the assignment itself.  We denote by $OPT(I)$ the minimum
number of bins that are assigned to all items in $I$.  We also denote
by $OPT$ an algorithm that for a given $I$, outputs such an assignment
that the number of non-empty bins is $OPT(I)$.

The definition of the asymptotic competitive ratio of an algorithm
$ALG$, denoted by $R_{ALG}$, is given by~(\ref{eq:none_acr}).  The
statements of the theorems in this section all take the form of an
inequality for $ALG(I)$ and $OPT(I)$. Bounds on the asymptotic
competitive ratio are obtained from the coefficients of $OPT(I)$.
\subsection{A Max-Min 1-Bounded Space Algorithm}
\label{subsec:none_maxmin_unit_upper}
%%
%% MM
%%
\begin{algorithm}
  \KwIn{An item sequence $I \in (0, 1]^{n}$}
  \KwOut{An assignment $f: \{1, 2, \ldots, n\} \to \mathbb{N}$}
  sort $I$ in a nonincreasing order and obtain $(a_{1}, a_{2}, \ldots, a_{n})$\;
  $h \leftarrow 1$, $t \leftarrow n$\;
  $l \leftarrow 1$, $S \leftarrow 0$\tcp*[r]{open a new bin}
  \While{$h \leq t$}{
    \While{$h \leq t$}{
      \uIf(\tcp*[f]{first try head}){$S + a_{h} \leq 1$}{
        $f(h) \leftarrow l$, $S \leftarrow S + a_{h}$\;
        $h \leftarrow h + 1$\;}
      \uElseIf(\tcp*[f]{then try tail}){$S + a_{t} \leq 1$}{
        $f(t) \leftarrow l$, $S \leftarrow S + a_{t}$\;
        $t \leftarrow t - 1$\;}
      \lElse(\tcp*[f]{both head and tail are too large}){break}
    }
    $l \leftarrow l + 1$, $S \leftarrow 0$\tcp*[r]{
      close the bin and open a new bin}
  }
  \Return{$f$}\;
  \caption{$MM$~\cite{ZHU201683}.}
  \label{algo:mm}
\end{algorithm}
Zhu proposed the following algorithm $MM$~\cite{ZHU201683}
(see Algorithm~\ref{algo:mm}
for the pseudocode):
First, sort the item sequence in a non-increasing order.  Then,
repeatedly assign the current single open bin first to as many of the
head items in the remaining sequence as possible and next to as many
of the tail items in the remaining sequence as possible.  Open a new
bin only when the current single open bin cannot accept the next tail
item.  If algorithm $MM$ has assigned a bin to the head (respectively,
tail) item in the currently remaining item sequence, then in the
resulting assignment we call that item a \emph{head} (respectively,
\emph{tail}) item.

The theorem below proves $R_{MM} \leq \frac{3}{2} = 1.5$.  The proof
is done either by an analysis of levels of bins or by an analysis
using weights as used
in~\cite{johnson1973near,10.1137/0602019,DBLP:journals/siamdm/Epstein06},
depending on whether an item of size at most $\frac{1}{3}$ is
contained in the latest bin in the assignment.
%%
%% See Appendix~\ref{subsec:thm:mm_cr} for the proof.
%%
%% 202411050904
%%
\begin{theorem}
  It holds that $MM(I) \leq \frac{3}{2} \cdot OPT(I) + 1$ for all item
  sequences $I$.
  \label{thm:mm_cr}
\end{theorem}
%%
%% NOAAPP
%%
\begin{proof}
  \begin{figure}[t]
    \centering
    \includegraphics[width=0.8\linewidth]{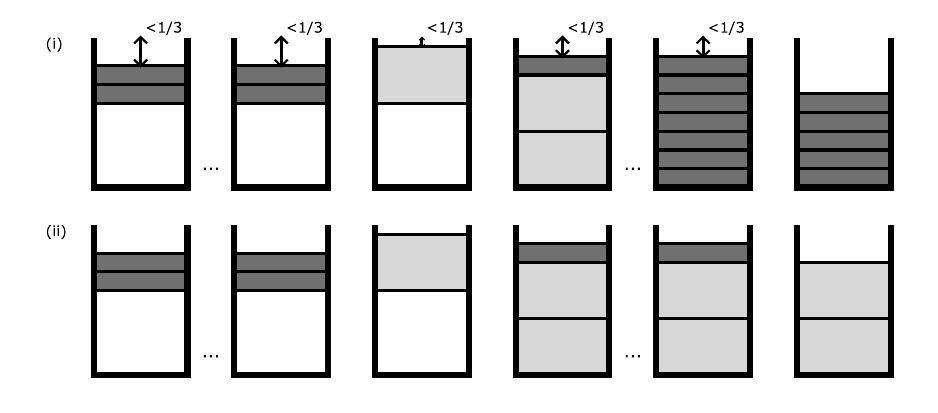}
    \caption{The assignment by $MM$.
      The white, light grey, and dark grey items are of
      classes 1, 2, and 3, respectively.}
    \label{fig:mm_cr}
  \end{figure}
  Without loss of generality, let $I = (a_{1}, a_{2}, \ldots, a_{n})
  \in (0, 1]^{n}$ be a sorted item sequence and assume that $I$ does
    not change by sorting by algorithm $MM$.
    In this proof, we call an item of size $(\frac{1}{2}, 1]$ a
      \emph{class~1} item, an item of size $(\frac{1}{3}, \frac{1}{2}]$
        a \emph{class~$2$} item, and an item of size $(0, \frac{1}{3}]$
          a \emph{class~$3$} item.
          
          (i)~The case where the latest bin in the assignment by
          algorithm $MM$ for $I$ contains at least one class~3 item
          (see Figure~\ref{fig:mm_cr}).  Since $I$ is sorted,
          throughout the execution of algorithm $MM$, the tail item in
          the currently remaining item sequence must always be a
          class~3 item.  Therefore, the sum of the sizes of the items
          of each bin, except the latest one, must exceed
          $\frac{2}{3}$. (If there were a bin which does not satisfy
          this, $MM$ would have packed one of the tail class~3 items
          additionally into the bin.)  Hence, by summing up the sizes
          of the items for all the bins other than the latest one, we
          get $\sum_{i=1}^{n} a_{i} > \frac{2}{3} \cdot (MM(I) -
          1)$.
          Since $\sum_{i=1}^{n} a_{i} \leq OPT(I)$, we have
          $MM(I) < \frac{3}{2} \cdot OPT(I) + 1$.
  
          (ii)~The case where the latest bin in the assignment by
          algorithm $MM$ for $I$ does not contain any class~3 item
          (see Figure~\ref{fig:mm_cr}).  The following weight function
          assigns a weight to each item according to the class:
          \begin{alignat*}{2}
            w_{1}(x) & =
            \begin{cases}
              1,           & x \in (\frac{1}{2}, 1];\\
                \frac{1}{2}, & x \in (\frac{1}{3}, \frac{1}{2}];\\
                  0,           & x \in (0, \frac{1}{3}].
            \end{cases}
          \end{alignat*}
          Since $I$ is sorted, throughout the execution of algorithm
          $MM$, the head item in the currently remaining item sequence
          must always be a class~1 or class~2 item.  Hence, there is
          no bins that contain only class~3 items in the assignment.
          Focus on the head items in the assignment.  Ignore the
          latest bin.  Then, due to the behavior of $MM$, each of the
          remaining bins turns out to be a bin containing one head
          class~1 item whose weight is $1$, or a bin containing two
          head class~2 items whose total weight is $\frac{1}{2} +
          \frac{1}{2} = 1$.  We thus obtain that the sum of the
          weights of the items in each of those bins is at least $1$.
          Therefore, $\sum_{i=1}^{n} w_{1}(a_{i}) \geq MM(I) - 1$
          holds in total.
          
          On the other hand, the sum of the weights of items that fit
          in a single bin is at most $\frac{3}{2}$, where the maximum
          is achieved by a class~1 item and a class~2 item. Therefore,
          $\sum_{i=1}^{n} w_{1}(a_{i}) \leq \frac{3}{2} \cdot OPT(I)$
          holds. We thus have $MM(I) \leq \frac{3}{2} \cdot OPT(I) +
          1$.
\end{proof}
\subsection{A Lower Bound for Max-Min Bounded Space Algorithms}
\label{subsec:none_maxmin_unit_lower}
%% %%
%% 202412030946
%%
%% \subsubsection{A Lower Bound for Max-Min 1-Bounded Space Algorithms}
\subsubsection{A Lower Bound for 1-Bounded Space Algorithms}
\label{subsubsec:none_maxmin_1_bounded_lower}
The next theorem states that $R_{ALG} \geq \frac{5}{4} = 1.25$ holds
for any max-min 1-bounded space algorithm $ALG$.  Compare with
algorithm $FFD$ (First Fit Decreasing) that achieves $R_{FFD} =
\frac{11}{9} (\approx 1.22)$~\cite{johnson1973near}.

$FFD$ first sorts the given item sequence in a non-increasing order
and then runs algorithm $FF$ (First Fit), that is, repeatedly packs
the head item in the currently remaining item sequence into the bin
with the smallest index that the item can fit in and opens a new bin
only when the item cannot fit in any open bin.  Thus, $FFD$ is a
pre-sorted online unbounded space algorithm.  The theorem reveals that
with only a single open bin, no max-min algorithm can outperform $FFD$
(see Table~\ref{tab:cr_without_cardi}).
%% See Appendix~\ref{subsec:thm:none_maxmin_unit_lower} for the proof.
%% The proof is omitted due to space limitations.
%% For lack of space, we omit the proof of Theorem~\ref{thm:none_maxmin_unit_lower}.
%%
\begin{theorem}
  For any max-min 1-bounded space algorithm $ALG$, and any positive
  integer $m$ divisible by 2, there exists a sorted item sequence $I$
  such that $OPT(I) = m$ and $ALG(I) \geq \frac{5}{4} \cdot OPT(I) -
  \frac{1}{4}$.
  \label{thm:none_maxmin_unit_lower}
\end{theorem}
\begin{proof}
  Fix a max-min 1-bounded space algorithm $ALG$ and a positive integer
  $m$ divisible by 2.  Choose $r$, $\varepsilon$, and $\delta$ such
  that $0 < r < \frac{\sqrt{5} - 1}{2}$, $0 < \varepsilon$, $0 <
  \delta$, $\varepsilon + \delta < \frac{1}{4} - \frac{1}{5} =
  \frac{1}{20}$, and $\delta r^{m-3} (1 - r - r^{2}) > 2 \varepsilon$.
  (For example, $r = \frac{\sqrt{3} - 1}{2}$, $\varepsilon =
  \frac{9r^{m-3}+10}{80(r^{m-3}+4)}$, and $\delta =
  \frac{13}{40(r^{m-3}+4)}$ satisfy the condition.)  Set $a_{i} =
  \frac{1}{2} + \delta r^{i-1}$ and $b_{i} = \frac{1}{4} - \varepsilon
  - \delta r^{i-1}$ for $i = 1, 2, \ldots, m$, as well as $c =
  \frac{1}{4} + \varepsilon$.  Take
  \begin{align*}
    I = (
    a_{1}, a_{2}, \ldots, a_{m},
    \overbrace{
      c, c, \ldots, c
    }^{m},
    b_{m}, b_{m-1}, \ldots, b_{1}).
  \end{align*}
  The items in $I$ are located in a non-increasing order.  Without
  loss of generality, assume that $I$ does not change by sorting by
  $ALG$.  Note that the order of the indices of $b_{\cdot}$-items is
  the reverse of their order of appearance.  Since $\frac{1}{2} <
  a_{i} < \frac{1}{2} + \delta < 1$, $\frac{1}{4} < c < \frac{1}{2} -
  \frac{1}{5} < \frac{1}{3}$, and $\frac{1}{5} < \frac{1}{4} -
  \varepsilon - \delta < b_{i} < \frac{1}{4}$, a single bin can
  contain at most one $a_{\cdot}$-item, at most three $c$-items, or at
  most four $b_{\cdot}$-items.

  To prove the theorem, the following property is important.
  \begin{equation}
    a_{i} + b_{i+1} + b_{i+2} > 1
    \text{ for } i = 1, 2, \ldots, m - 2.
    \label{eq:none_items_overflow}
  \end{equation}
  This can be shown as follows: $0 < r < \frac{\sqrt{5} - 1}{2}$
  implies $1 - r - r^{2} > 0$. Thus, for $i = 1, 2, \ldots, m - 2$, we
  get $a_{i} + b_{i+1} + b_{i+2} = 1 - 2 \varepsilon + \delta r^{i-1}
  (1 - r - r^{2}) \geq 1 - 2 \varepsilon + \delta r^{m-3} (1 - r -
  r^{2}) > 1$.

  It is seen that $a_{i} + c + b_{i} = 1$ for $i = 1, 2, \ldots, m$,
  which means that if the $a_{i}$-item, a $c$-item, and the $b_{i}$-item
  are packed into a bin, the bin is just full. Therefore, $OPT(I) = m$
  holds.

  Since $ALG$ is a max-min algorithm, $ALG$ repeatedly packs either
  the head or tail item in the currently remaining item sequence into
  a bin.  Besides, since $ALG$ is a 1-bounded space algorithm, $ALG$
  packs each item always into the latest bin, not into any earlier
  bin.  We denote by $D_{i}$ the number of $b_{\cdot}$-items that are
  packed by $ALG$ together with $a_{1}, a_{2}, \ldots, a_{i}$ for $i =
  1, 2, \ldots, m$. The following holds:
  \begin{equation}
    D_{i} \leq i + 1
    \text{ for } i = 1, 2, \ldots, m - 1.
    \label{eq:none_few_small_items}
  \end{equation}
  This is shown by induction.  Since $a_{1} = \frac{1}{2} + \delta$
  and $b_{i} = \frac{1}{4} - \varepsilon - \delta r^{i-1}$,
  the $a_{1}$-item and at most two $b_{\cdot}$-items can be packed
  together.  Thus, $D_{1} \leq 2$ holds.  Next, assume $D_{j-1} \leq
  j$ for some $j \geq 2$.  If $D_{j-1} \leq j - 1$, then $D_{j} \leq
  D_{j-1} + 2 \leq j - 1 + 2 = j+1$ immediately holds since
  the $a_{j}$-item and at most two $b_{\cdot}$-items can be packed
  together.  If $D_{j-1} = j$, then the maximum index of the
  $b_{\cdot}$-items that can be packed together with $a_{1}, a_{2},
  \ldots, a_{j-1}$ is at least $j$.  Since $ALG$ assigns bins to
  $a_{1}, a_{2}, \ldots, a_{j}$ along this order, the index of the
  $b_{\cdot}$-item that is packed together with the $a_{j}$-item is at
  least $j+1$.  By~(\ref{eq:none_items_overflow}), we have $a_{j} +
  b_{j'} + b_{j'+1} > a_{j} + b_{j+1} + b_{j+2} > 1$. Hence, at most
  one $b_{\cdot}$-item can be packed together with the $a_{j}$-item.  We
  thus have $D_{j} \leq D_{j-1} + 1 \leq j + 1$.
  In this way~(\ref{eq:none_few_small_items}) is proved.
  
  \begin{figure}[t]
    \centering
    \includegraphics[width=0.9\linewidth]{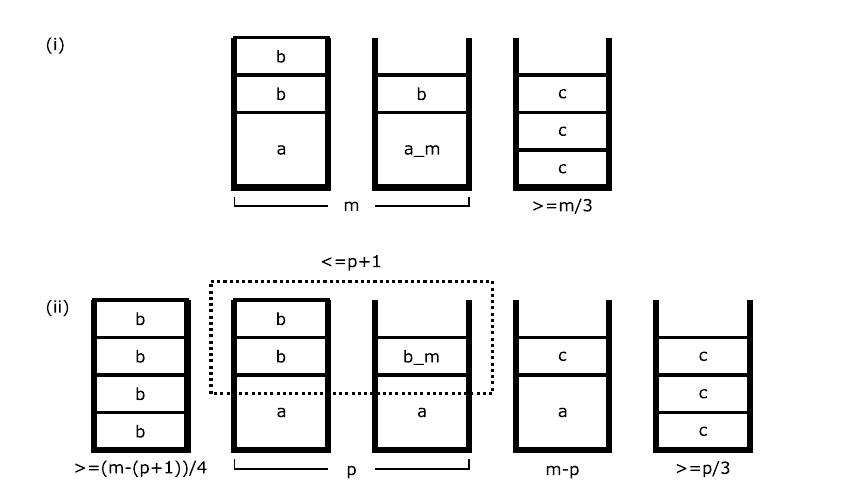}
    \caption{The assignment by $ALG$ for $I$.}
    \label{fig:none_maxmin_unit_lower}
  \end{figure}
  Let us evaluate $ALG(I)$ (see
  Figure~\ref{fig:none_maxmin_unit_lower}).
  
  (i)~The case where $ALG$
  packs the $a_{m}$-item before the $b_{m}$-item.  When the $a_{m}$-item has just
  been packed, all $m$ of $c$-items must remain unprocessed.  Since at
  most one $a_{\cdot}$-item can fit in a bin, the $a_{1}$-item through
  the $a_{m}$-item occupy exactly $m$ bins.  Since at most three $c$-items
  can fit in a bin, $m$ of $c$-items use more than $\frac{m}{3}$ bins.
  Therefore, we obtain
  \begin{equation*}
    ALG(I) \geq m + \frac{m}{3} = \frac{4}{3} m.
  \end{equation*}

  (ii)~If $ALG$ packs the $b_{m}$-item before the $a_{m}$-item. Let $p$ be the
  maximum index of the $a_{\cdot}$-items that are contained in the bin
  containing the $b_{m}$-item or in earlier bins than it.  If there is no
  such $a_{\cdot}$-items, let $p = 0$.
  By~(\ref{eq:none_few_small_items}), there are at most $(p + 1)$ of
  $b_{\cdot}$-items packed together with $a_1, a_2, \ldots, a_p$.
  Thus, in order to pack the remaining $b_{\cdot}$-items, at least
  $\frac{m - (p + 1)}{4}$ bins are needed.

  On the other hand, when
  the $b_{m}$-item has just been packed, $(m - p)$ of $a_{\cdot}$-items and $m$ of
  $c$-items are all still unprocessed. $(m - p)$ of $a_{\cdot}$-items
  require $(m - p)$ bins. Since one $a_{\cdot}$-item and at most
  one $c$-item can fit in a bin, at most $(m - p)$ of
  $c$-items are packed together with $a_{\cdot}$-items. To pack the
  remaining $c$-items, $\frac{m - (m - p)}{3} = \frac{p}{3}$ or more bins
  should be required. Therefore, we get
  \begin{equation*}
    ALG(I) \geq \frac{m - (p + 1)}{4} + p + (m - p) + \frac{p}{3}
    =    \frac{5}{4} m + \frac{1}{12} (p - 3)
    \geq \frac{5}{4} m - \frac{1}{4}.
  \end{equation*}

  The theorem is proved since $\frac{4}{3} m > \frac{5}{4} m - \frac{1}{4}$.
\end{proof}
%%
%% AAPP
%%
%% \subsection{A Lower Bound for Max-Min Bounded Space Algorithms}
%% \label{subsec:none_maxmin_bounded_lower}
%%
%% \subsubsection{A Lower Bound for Max-Min $B$-Bounded Space Algorithms}
\subsubsection{A Lower Bound for $B$-Bounded Space Algorithms}
\label{subsubsec:none_maxmin_b_bounded_lower}
%%
%% For bounded space algorithms,
We obtain a weaker lower bound than
Theorem~\ref{thm:none_maxmin_unit_lower}.  More specifically, we
present Theorem~\ref{thm:none_maxmin_bounded_lower}, which states that
$R_{ALG} \geq \frac{7}{6} (\approx 1.17)$ holds for any max-min
bounded space algorithm $ALG$.
%%
%% The proof is omitted due to space limitations.
%% The proof is omitted due to lack of space.
%% See Appendix~\ref{subsec:thm:none_maxmin_bounded_lower} for the proof.
%%
%% 202411261359
%%
\begin{theorem}
  For any max-min $B$-bounded space algorithm $ALG$ and any positive
  integer $m$ that is at least $2 B$ and divisible by 3, there exists
  a sorted item sequence $I$ such that $OPT(I) = m$ and $ALG(I) \geq
  \frac{7}{6} \cdot OPT(I) - B$.
  \label{thm:none_maxmin_bounded_lower}
\end{theorem}
\begin{proof}
  Fix a max-min $B$-bounded space algorithm $ALG$ and a positive
  integer $m$ that is at least $2 B$ and divisible by 3.  Choose
  $\varepsilon$ such that $0 < \varepsilon < \frac{1}{2} \cdot
  (\frac{1}{6} - \frac{1}{7}) = \frac{1}{84}$.  Take
  \begin{equation*}
    I = \Bigl(
    \overbrace{
      \frac{1}{2} + \varepsilon, \ldots, \underline{\frac{1}{2} + \varepsilon}
    }^{m},
    \overbrace{
      \frac{1}{3} + \varepsilon, \ldots, \frac{1}{3} + \varepsilon
    }^{m},
    \overbrace{
      \underline{\frac{1}{6} - 2 \varepsilon},
      \ldots, \frac{1}{6} - 2 \varepsilon
    }^{m} \Bigr).
  \end{equation*}
  Here, the last $(\frac{1}{2} + \varepsilon)$-item and the first
  $(\frac{1}{6} - 2 \varepsilon)$-item are important, so we underlined
  these items.  The items in $I$ are located in a non-increasing
  order.  Without loss of generality, assume that $I$ does not change
  by sorting by $ALG$.
  
  Since the sum of the item sizes of a $(\frac{1}{2} +
  \varepsilon)$-item, a $(\frac{1}{3} + \varepsilon)$-item, and a
  $(\frac{1}{6} - 2 \varepsilon)$-item is exactly 1, $OPT(I) = m$
  holds.

  \begin{figure}[t]
    \centering
    \includegraphics[width=0.8\linewidth]{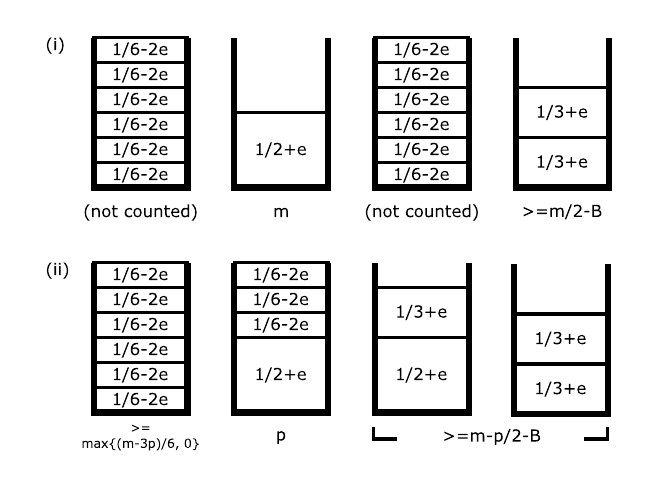}
    \caption{The assignment by $ALG$ for $I$.
    }
    \label{fig:none_maxmin_bounded_lower}
  \end{figure}
  Let us evaluate $ALG(I)$ (see Figure~\ref{fig:none_maxmin_bounded_lower}).
  Since $ALG$ is a max-min algorithm, $ALG$
  repeatedly packs either the head or
  tail item in the currently remaining item sequence into a bin.

  (i)~The case where $ALG$ packs the last $(\frac{1}{2} +
  \varepsilon)$-item before the first $(\frac{1}{6} - 2
  \varepsilon)$-item.  
  Consider the moment
  immediately after $ALG$ has packed the last $(\frac{1}{2} +
  \varepsilon)$-item.  
  Since at most one
  $(\frac{1}{2} + \varepsilon)$-item can fit in a bin, the number
  of bins assigned so far is at least $m$. At most $B$
  of these bins remain open.  Next, consider all the unprocessed items
  at this moment.
  Ignore $(\frac{1}{6} - 2 \varepsilon)$-items
  and see only $(\frac{1}{3} + \varepsilon)$-items. Since at most two
  $(\frac{1}{3} + \varepsilon)$-items can fit in a bin and there
  are at most $B$ bins, the number of bins that are needed to pack $(\frac{1}{3} +
  \varepsilon)$-items is at least $(\frac{m}{2} - B)$, which is
  nonnegative because we have assumed $m \geq 2 B$.  From above derive
  \begin{equation*}
    ALG(I) \geq m + \frac{m}{2} - B =
    \frac{3}{2} m - B.
  \end{equation*}

  (ii)~The case where $ALG$ packs the first $(\frac{1}{6} - 2
  \varepsilon)$-item before the last $(\frac{1}{2} +
  \varepsilon)$-item.
  When $ALG$ has packed the first $(\frac{1}{6} - 2
  \varepsilon)$-item,
  $ALG$ should have packed $m$ of $(\frac{1}{6} - 2 \varepsilon)$-items.
  Let $p$ be the number of $(\frac{1}{2} +
  \varepsilon)$-items that $ALG$ has packed before the first
  $(\frac{1}{6} - 2 \varepsilon)$-item.  If there is no
  such $(\frac{1}{2} + \varepsilon)$-items, let $p = 0$.  At most one
  $(\frac{1}{2} + \varepsilon)$-item can fit in a bin.
  Also, each of the bins containing a
  $(\frac{1}{2} + \varepsilon)$-item can additionally
  contain at
  most 3 of $(\frac{1}{6} - 2 \varepsilon)$-items.  Thus, the number
  of bins assigned before the first $(\frac{1}{6} - 2
  \varepsilon)$-item gets packed is $\max \{ \frac{m - 3 p}{6}, 0\} +
  p$. At most $B$ of these bins remain open.

  When the first $(\frac{1}{6} - 2 \varepsilon)$-item gets packed, the
  remaining items are $(m - p)$ of $(\frac{1}{2} + \varepsilon)$-items
  and $m$ of $(\frac{1}{3} + \varepsilon)$-items.  Since $m - p \leq
  m$, these items can be packed at most two per bin.  Therefore, the
  number of bins needed to pack these items, making use of the open
  bins, is at least $\frac{(m - p) + m}{2} - B = m - \frac{p}{2} - B$,
  which is nonnegative because we have assumed $m \geq 2 B$.  Hence,
  we evaluate
  \begin{equation*}
    ALG(I) \geq
    \max \Bigl\{ \frac{m - 3 p}{6}, 0\Bigr\} + p + \frac{(m - p) + m}{2} - B
    = \max \Bigl\{ \frac{m}{6}, \frac{p}{2}\Bigr\} + \frac{m}{2} - B
    \geq \frac{7}{6} m - B.
  \end{equation*}
  The above evaluation is obtained since the max operator takes the
  minimum value $\frac{m}{6}$ when $p = \frac{m}{3}$, that is,
  the operands are equal. This is always valid since we have
  assumed that $m$ is divisible by $3$.

  The theorem is proved since $\frac{3}{2} m - B > \frac{7}{6} m - B$.
\end{proof}
%%
%% AAPP
%%
\subsection{A Lower Bound for Max-Min Unbounded Space Algorithms}
\label{subsec:none_maxmin_unbounded_lower}
We establish, as Theorem~\ref{thm:none_maxmin_unbounded_lower}, an
even weaker lower bound for unbounded space algorithms: $R_{ALG} \geq
\frac{16}{15} (\approx 1.07)$ holds for any max-min unbounded space
algorithm $ALG$.  The proof involves a technique of suspending the
execution of the algorithm when it packs special items in order to
construct another item sequence depending on the behavior of the
algorithm.  This exploits the weakness of the max-min algorithm in
that it can only see the head and tail items of the currently
remaining item sequence.
%%
%% The proof is omitted due to lack of space.
%% See Appendix~\ref{subsec:thm:none_maxmin_unbounded_lower} for the
%% proof.
%%
%% 202501200844
%%
\begin{theorem}
  For any max-min unbounded space algorithm $ALG$, and any positive
  integer $m$, there exists a sorted item sequence $I$ such that
  $OPT(I) > m$ and $ALG(I) > \frac{16}{15} \cdot (OPT(I) - 1)$.
  \label{thm:none_maxmin_unbounded_lower}
\end{theorem}
\begin{proof}
  Fix a max-min unbounded space algorithm $ALG$ and a positive integer
  $m$.  Set $N = 4m$.  Choose $\varepsilon$ such that $0 < \varepsilon
  < \frac{1}{36}$.  Take
  \begin{align*}
    I_{+} & = \Bigl(
    \overbrace{
      \frac{1}{3} + 2 \varepsilon, \ldots,
      \underline{\frac{1}{3} + 2 \varepsilon}
    }^{N},
    \overbrace{
      \frac{1}{3} + \varepsilon, \ldots, \frac{1}{3} + \varepsilon
    }^{N},
    \overbrace{
      \underline{\frac{1}{3} - 3 \varepsilon},
      \ldots, \frac{1}{3} - 3 \varepsilon
    }^{N}
    \Bigr).
  \end{align*}
  Here, we underlined the last $(\frac{1}{3} + 2 \varepsilon)$-item
  and the first $(\frac{1}{3} - 3 \varepsilon)$-item, which play an
  important role.  The items in $I_{+}$ are located in a
  non-increasing order.  Without loss of generality, assume that
  $I_{+}$ does not change by sorting by $ALG$.

  Since the sum of the sizes of a $(\frac{1}{3} + 2
  \varepsilon)$-item, a $(\frac{1}{3} + \varepsilon)$-item, and a
  $(\frac{1}{3} - 3 \varepsilon)$-item is exactly 1, we have
  $OPT(I_{+}) = N > m$.

  \begin{figure}[t]
    \centering
    \includegraphics[width=1.0\linewidth]
                    {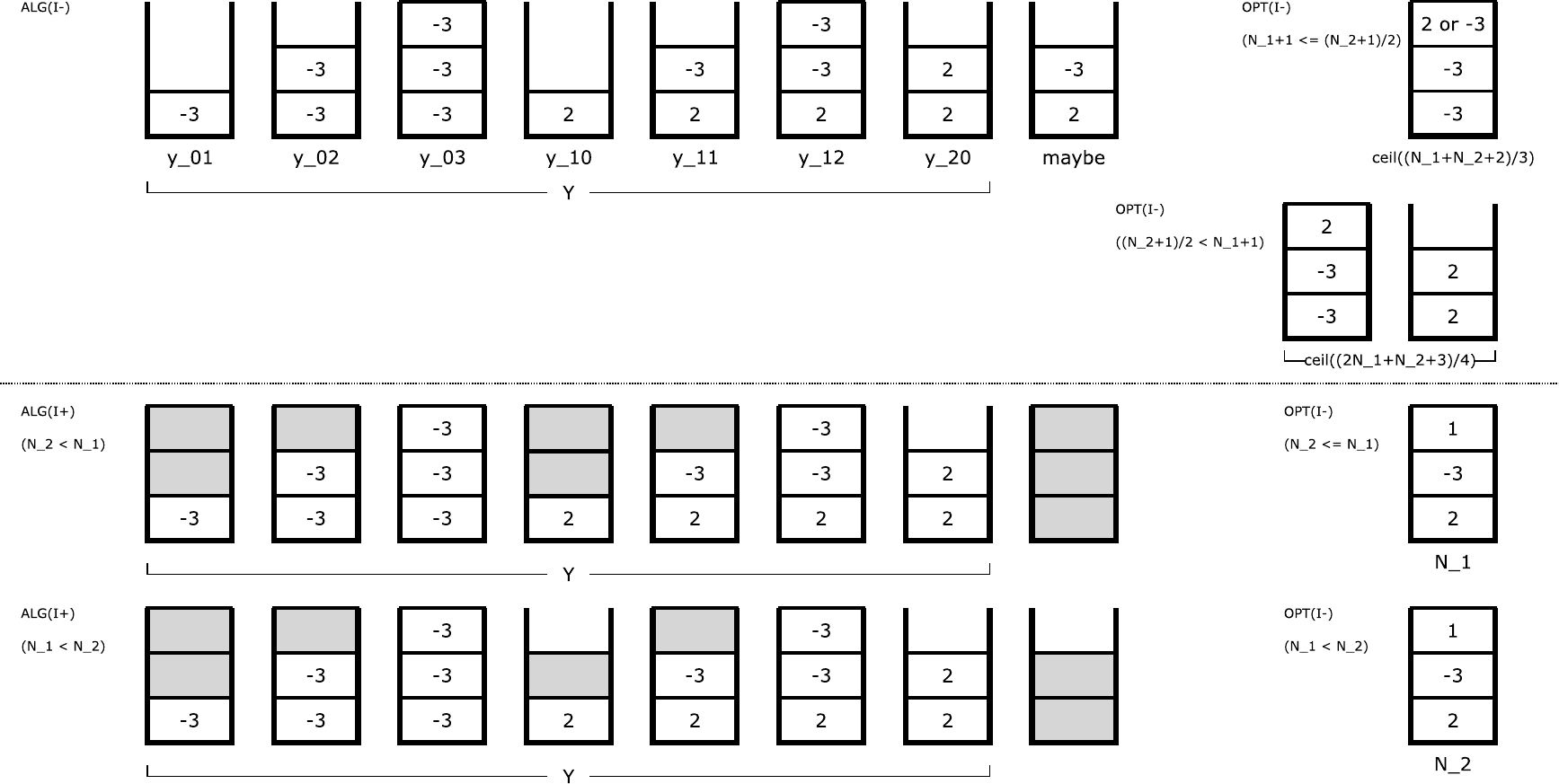}
    \caption{The assignments by $ALG$ and $OPT$ for $I_{-}$ and $I_{+}$.
      An item with label $a$ is of size $\frac{1}{3} + a \varepsilon$.
      The light grey items are the difference of $I_{+}$ and $I_{-}$.
    }
    \label{fig:none_maxmin_unbounded_lower}
  \end{figure}
  (I)~By the definition of max-min algorithms, $ALG$
  repeatedly packs either the head or
  tail item in the currently remaining item sequence into a bin.
  Start $ALG$ for $I_{+}$ and interrupt the execution immediately
  after $ALG$ has packed \emph{either} of the underlined items in
  $I_{+}$.  (The result of the whole run will be analyzed later
  in~(II).)  Among the items in $I_{+}$ that have been packed so far,
  let $N_{1}$ be the number of $(\frac{1}{3} + 2 \varepsilon)$-items,
  and $N_{2}$ be the number of $(\frac{1}{3} - 3 \varepsilon)$-items.
  Obviously, $\max \{N_{1}, N_{2}\} = N$.

  Now, consider another item sequence determined by $N_{1}$ and $N_{2}$:
  \begin{align*}
    I_{-} & = \Bigl(
    \overbrace{
      \frac{1}{3} + 2 \varepsilon, \ldots,
      \underline{\underline{\frac{1}{3} + 2 \varepsilon}},
      \frac{1}{3} + 2 \varepsilon
    }^{N_{1} + 1},
    \overbrace{
      \frac{1}{3} - 3 \varepsilon,
      \underline{\underline{\frac{1}{3} - 3 \varepsilon}}, \ldots,
      \frac{1}{3} - 3 \varepsilon
    }^{N_{2} + 1}
    \Bigr).
  \end{align*}
  The double underlined items here are the $N_{1}$-th $(\frac{1}{3} +
  2 \varepsilon)$-item from the head and the $N_{2}$-th $(\frac{1}{3}
  - 3 \varepsilon)$-item from the tail.  Without loss of generality,
  assume that $I_{-}$ does not change by sorting by $ALG$.

  %% 202506242113

  Start $ALG$ for $I_{-}$ and interrupt the execution immediately
  after $ALG$ has packed \emph{both} of the double underlined
  items in $I_{-}$.
  Then, the two items between the double underlined
  items are
  necessarily left unprocessed.
  The reason is confirmed as follows.  Assume that $N_{1} = N$.  (The
  %% The reason is why this happens
  %% confirmed as follows.  Assume that $N_{1} = N$.  (The
  case of $N_{2} = N$ can be discussed in the same way.)  Then, the
  last item that is packed before the interruption $ALG$ for $I_{+}$
  %% is the $N_{1}$-th item from the head.  Consider the moment
  %% is the $N_{1}$-th item from the head.  Look at the moment
  is the $N_{1}$-th item from the head.  Go back to the moment
  immediately before this item is packed.  At this moment, for
  $I_{+}$, $ALG$ has packed $(N_{1} - 1)$ items from the head and
  $N_{2}$ items from the tail.  The same situation should arise in the
  execution of $ALG$ for $I_{-}$ as well, since the items that $ALG$
  sees for $I_{-}$ are exactly the same as those seen for $I_{+}$.
  Specifically, during the execution for both $I_{-}$ and $I_{+}$,
  the head and tail items in the
  currently remaining item sequence have remained a $(\frac{1}{3} + 2
  \varepsilon)$-item and a $(\frac{1}{3} - 3 \varepsilon)$-item,
  %% respectively.  Hence, also for $I_{-}$, the next item $ALG$ packs is the
  respectively.  Hence, also for $I_{-}$, $ALG$ next packs is the
  $N_{1}$-th item from the head.  Now both the double underlined items
  have been packed and the execution is interrupted, leaving the two
  items between the double underlined items unprocessed.

  See the assignment after the interruption of the execution of $ALG$
  for $I_{-}$, which should be identical to
  the assignment obtained by interrupting $ALG$ for $I_{+}$.
  We refer to a bin that contains $i$ of $(\frac{1}{3} + 2
  \varepsilon)$-items and $j$ of $(\frac{1}{3} - 3 \varepsilon)$-items
  as an \emph{$ij$-bin}. Let $y_{ij}$ be the number of $ij$-bins in
  the assignment.
  The number of non-empty bins in the assignment is
  \begin{equation}
    Y :=
    y_{01} + y_{02} + y_{03} + y_{10} + y_{11} + y_{12} + y_{20}
    \label{eq:alg_i_minus}
  \end{equation}
  (see Figure~\ref{fig:none_maxmin_unbounded_lower}).
  $N_{1}$ of $(\frac{1}{3} + 2 \varepsilon)$-items and
  $N_{2}$ of $(\frac{1}{3} - 3 \varepsilon)$-items are packed somewhere.
  Therefore, it should hold that
  \begin{equation}
    y_{10} + y_{11} + y_{12} + 2 y_{20} - N_{1}
    = 0 \label{eq:2e_items}
  \end{equation}
  and
  \begin{equation}
    y_{01} + 2 y_{02} + 3 y_{03} + y_{11} + 2 y_{12} - N_{2}
    = 0. \label{eq:3e_items}
  \end{equation}

  Let $ALG$ resume the execution for $I_{-}$ and pack the remaining
  two items in $I_{-}$.  It is immediate that
  \begin{equation}
    Y \leq ALG(I_{-}).
    \label{eq:alg_i_minus_bound}
  \end{equation}

  Next, evaluate $OPT(I_{-})$ (see
  Figure~\ref{fig:none_maxmin_unbounded_lower}).  If $N_{1} + 1 \leq
  \frac{N_{2} + 1 }{2}$, it is optimal to create as many
  bins containing $3$ of $(\frac{1}{3} - 3 \varepsilon)$-items in them or
  bins containing $2$ of $(\frac{1}{3} - 3 \varepsilon)$-items and with $1$ of
  $(\frac{1}{3} + 2 \varepsilon)$-item as possible. Thus, we have
  \begin{equation}
    OPT(I_{-}) = \Bigl\lceil \frac{(N_{1} + 1) + (N_{2} + 1)}{3} \Bigr\rceil
    \geq \frac{N_{1} + N_{2} + 2}{3} > \frac{N}{3} > m.
    \label{eq:opt_i_minus_1}
  \end{equation}
  On the other hand,
  if $\frac{N_{2} + 1}{2} < N_{1} + 1$, it is optimal first to create
  as many bins containing $2$ of $(\frac{1}{3} - 3
  \varepsilon)$-items as possible, next to add 
  $(\frac{1}{3} + 2 \varepsilon)$-items to these bins, and finally
  to create as many bins containing $2$ of $(\frac{1}{3} + 2 \varepsilon)$-items.
  Then, we have
  \begin{equation}
    OPT(I_{-}) = 
    \Bigl\lceil \frac{N_{2} + 1}{2} + \frac{(N_{1} + 1) - \frac{N_{2} + 1}{2}}{2}
    \Bigr\rceil = \Bigl\lceil \frac{2 N_{1} + N_{2} + 3}{4} \Bigr\rceil
    > \frac{N_{1} + N_{2}}{4}
    \geq \frac{N}{4} = m.
    \label{eq:opt_i_minus_2}
  \end{equation}
  Note that we have applied $N_{1} + N_{2} \geq \max \{N_{1}, N_{2}\} = N$.

  (II)~To prove the theorem, it is sufficient that at least one of
  $ALG(I_{-}) > \frac{16}{15} \cdot (OPT(I_{-}) - 1)$ or $ALG(I_{+}) >
  \frac{16}{15} \cdot (OPT(I_{+}) - 1)$ holds. If $ALG(I_{-}) >
  \frac{16}{15} \cdot (OPT(I_{-}) - 1)$, the proof is 
  complete. Therefore, from now on, assume
  \begin{equation}
    Y \leq
    ALG(I_{-}) \leq \frac{16}{15} \cdot (OPT(I_{-}) - 1),
    \label{eq:alg_i_minus_assumption}
  \end{equation}
  where we have applied~(\ref{eq:alg_i_minus_bound}).
  In the rest of the proof, we are going to show
  $ALG(I_{+}) > \frac{16}{15} \cdot (OPT(I_{+}) - 1)$.
  Let $ALG$ resume the execution for $I_{+}$ and pack all the
  remaining items in $I_{+}$ (see
  Figure~\ref{fig:none_maxmin_unbounded_lower}).
  Set $r = \frac{ALG(I_{+})}{OPT(I_{+})}$.

  (i)~The case where
  \begin{equation}
    N_{2} - N_{1} < 0. \label{eq:case_ro}
  \end{equation}
  The remaining items in $I_{+}$ are $(N_{1} - N_{2})$ of
  $(\frac{1}{3} - 3 \varepsilon)$-items and $N_{1}$ of $(\frac{1}{3} +
  \varepsilon)$-items. See how many of these items can be packed into
  the bins that already contain items
  (see Figure~\ref{fig:none_maxmin_unbounded_lower}).
  Due to the size
  constraint, the $20$-bins can no longer hold any items. The $02$-bins
  and the $11$-bins can each hold just one item. The $01$-bins and the
  $10$-bins can each hold at most two items. As for the further
  remaining items, at most three of them can fit in a bin.
  Thus overall, we evaluate
  \begin{equation*}
    ALG(I_{+}) \geq Y +
    \max\Bigl\{
    \frac{(N_{1} - N_{2}) + N_{1} - 2 y_{01} - y_{02} - 2 y_{10} - y_{11}}{3}
    , 0 \Bigr\}.
  \end{equation*}
  Applying $r = \frac{ALG(I_{+})}{OPT(I_{+})}$, we obtain
  \begin{equation}
    Y +
    \frac{(N_{1} - N_{2}) + N_{1} - 2 y_{01} - y_{02} - 2 y_{10} - y_{11}}{3}
    - r N_{1}  \leq 0. \label{eq:alg_i_minus_i_1}
  \end{equation}

  From~(\ref{eq:opt_i_minus_2}) and~(\ref{eq:alg_i_minus_assumption}),
  it is derived that
  \begin{equation}
    Y - \frac{16}{15} \cdot
    \frac{2 N_{1} + N_{2} + 3}{4}
    < 0. \label{eq:alg_i_minus_16_15_a}
  \end{equation}
  Calculating
  $\frac{4}{45} \cdot (\ref{eq:case_ro})
  + 1 \cdot (\ref{eq:alg_i_minus_i_1})
  + \frac{1}{3} \cdot (\ref{eq:alg_i_minus_16_15_a})
  - \frac{2}{3} \cdot (\ref{eq:2e_items})
  - \frac{1}{3} \cdot (\ref{eq:3e_items})$
  and eliminating $Y$ by~(\ref{eq:alg_i_minus}), we get
  \begin{equation}
    \frac{y_{01} + y_{02} + y_{03}}{3}
    + \frac{16}{15} N_{1} -\frac{4}{15}
    < r N_{1}.
    \label{eq:ro_lb}
  \end{equation}
  The fact that $y_{01}$, $y_{02}$, and $y_{03}$ are all nonnegative,
  $OPT(I_{+}) = N_{1}$, and $r = \frac{ALG(I_{+})}{OPT(I_{+})}$ leads us to
  that $ALG(I_{+}) > \frac{16}{15} \cdot OPT(I_{+}) -
  \frac{4}{15} > \frac{16}{15} \cdot (OPT(I_{+}) - 1)$. 
  (Actually,
  this derivation corresponds to solving a linear optimization problem of
  minimizing $r$ under the constraints~(\ref{eq:case_ro}),
  (\ref{eq:alg_i_minus_i_1}), (\ref{eq:alg_i_minus_16_15_a}),
  (\ref{eq:2e_items}), and~(\ref{eq:3e_items}).)

  (ii)~The case where
  \begin{equation}
    N_{1} - N_{2} < 0. \label{eq:case_yi}
  \end{equation}
  The remaining items in $I_{+}$ are $(N_{2} - N_{1})$ of
  $(\frac{1}{3} + 2 \varepsilon)$-items and $N_{2}$ of $(\frac{1}{3}
  + \varepsilon)$-items.
  Similarly to~(i), see how many of these items can be packed
  into the bins that already contain items
  (see Figure~\ref{fig:none_maxmin_unbounded_lower}).
  Due to the size
  constraint, the $20$-bins can no longer hold any items. The $02$-bins,
  the $10$-bin, and the $11$-bins
  can each hold just one item. The $01$-bins can
  each hold at most two items.
  As for the further remaining items,
  at most two of them can fit in a bin.
  Therefore, we have
  \begin{equation*}
    ALG(I_{+}) \geq Y +
    \max\Bigl\{
    \frac{(N_{2} - N_{1}) + N_{2} - 2 y_{01} - y_{02} - y_{10} - y_{11}}{2}
    , 0 \Bigr\}.
  \end{equation*}
  Using $r = \frac{ALG(I_{+})}{OPT(I_{+})}$, we get
  \begin{equation}
    Y +
    \frac{(N_{2} - N_{1}) + N_{2} - 2 y_{01} - y_{02} - y_{10} - y_{11}}{2}
    - r N_{2}  \leq 0. \label{eq:alg_i_minus_ii_1}
  \end{equation}
  We perform further case analysis for $OPT(I_{-})$.

  (ii-a)~The case where
  $N_{1} < N_{2}$ and $\frac{N_{2} + 1}{2} < N_{1} + 1$, that is,
  $\frac{N_{2} - 1}{2} < N_{1} < N_{2}$.  In this case,
  (\ref{eq:alg_i_minus_16_15_a}) is valid.
  Calculating
  $\frac{1}{30} \cdot (\ref{eq:case_yi})
  + 1 \cdot (\ref{eq:alg_i_minus_ii_1})
  + 1 \cdot (\ref{eq:alg_i_minus_16_15_a})
  - 1 \cdot (\ref{eq:2e_items})
  - \frac{1}{2} \cdot (\ref{eq:3e_items})$
  and eliminating $Y$ by~(\ref{eq:alg_i_minus}), we derive
  \begin{equation}
    \frac{y_{01} + y_{02} + y_{03} + y_{10}}{2}
    + \frac{6}{5} N_{2} - \frac{4}{5}
    < r N_{2}.
    \label{eq:yi_a_lb}
  \end{equation}
  Applying the fact that $y_{01}$, $y_{02}$, $y_{03}$, and $y_{10}$ are all
  nonnegative, $OPT(I_{+}) = N_{2}$, and $r =
  \frac{ALG(I_{+})}{OPT(I_{+})}$, we conclude
  $ALG(I_{+}) > \frac{6}{5} \cdot
  OPT(I_{+}) - \frac{4}{5} > \frac{16}{15} \cdot (OPT(I_{+}) - 1)$.

  (ii-b)~The case where
  \begin{equation}
    N_{1} - \frac{N_{2} - 1}{2} \leq 0. \label{eq:case_yi_b}
  \end{equation}
  (\ref{eq:opt_i_minus_1}) and~(\ref{eq:alg_i_minus_assumption}) imply
  that
  \begin{equation}
    Y - \frac{16}{15} \cdot
    \frac{N_{1} + N_{2}}{3}
    < 0. \label{eq:alg_i_minus_16_15_b}
  \end{equation}
  Calculating
  $\frac{8}{45} \cdot (\ref{eq:case_yi_b})
  + 1 \cdot (\ref{eq:alg_i_minus_ii_1})
  + \frac{1}{2} \cdot (\ref{eq:alg_i_minus_16_15_a})
  - \frac{1}{2} \cdot (\ref{eq:2e_items})
  - \frac{1}{2} \cdot (\ref{eq:3e_items})$
  and eliminating $Y$ by~(\ref{eq:alg_i_minus}), we obtain
  \begin{equation}
    \frac{y_{01} + y_{02}}{2} + \frac{37}{30} N_{2} - \frac{4}{15}
    < r N_{2}.
    \label{eq:yi_b_lb}
  \end{equation}
  Since $y_{01}$ and $y_{02}$ are nonnegative,
  $OPT(I_{+}) = N_{2}$, and $r = \frac{ALG(I_{+})}{OPT(I_{+})}$,
  it follows that
  $ALG(I_{+}) > \frac{37}{30} \cdot OPT(I_{+}) - \frac{4}{15}
  > \frac{16}{15} \cdot (OPT(I_{+}) - 1)$.
\end{proof}

%% AAPP
%%
\subsection{A Tight Bound for Pre-Sorted Online Bounded Space Algorithms}
\label{subsec:none_online_sorted_bounded_tight}
%%
%% 202412261543
%%
Recall the sequence $\{\pi_{i}\}_{i = 1}^{\infty}$ defined by $\pi_{1}
= 2$ and $\pi_{i+1} = \pi_{i}(\pi_{i} - 1) + 1$ for $i \geq 1$ in
Section~\ref{subsec:ourcontribution}.  The values are: $\pi_{1} = 2$,
$\pi_{2} = 3$, $\pi_{3} = 7$, $\pi_{4} = 43$, $\pi_{5} = 1807$, and so
on.  By induction, we can see that $\{\pi_{i}\}_{i = 1}^{\infty}$ is
monotonically increasing and
\begin{equation}
  1 - \sum_{i=1}^{j} \frac{1}{\pi_{j}} = \frac{1}{\pi_{j+1} - 1}
  \text{ for all } j \geq 1.
  \label{eq:pi_sum}
\end{equation}
We also present the following convergence:
\begin{equation}
  \sum_{i=1}^{\infty} \frac{1}{\pi_{i}-1} = \gamma  (\approx 1.69).
  \label{eq:gamma_converge}
\end{equation}

Lee and Lee provided a tight bound of $\gamma$ on the asymptotic
approximation ratio of online bounded space
algorithms~\cite{Lee:1985:SOB:3828.3833}.  In this section, by
combining with an existing result, we provide a tight bound for
pre-sorted online bounded space algorithms, which is also equal to
$\gamma$.  This reveals that pre-sorting gives no advantage for online
bounded space algorithms.

Algorithm $NFD$ (Next Fit Decreasing) first sorts the given item
sequence in a non-increasing order and then runs algorithm $NF$ (Next
Fit), that is, repeatedly packs, into the current single open bin, as
many of the head items in the remaining item sequence as possible and
opens a new bin only when the item cannot fit in the single open bin.
It is seen that $NFD$ is a pre-sorted online 1-bounded space
algorithm.
%% It is seen that $NFD$ is one of pre-sorted online 1-bounded space
%% algorithms.

Baker and Coffman Jr.\ proved $R_{NFD} =
\gamma$~\cite{10.1137/0602019}, which we introduce as
Theorem~\ref{thm:none_online_sorted_bounded_upper}.  Together with
Theorem~\ref{thm:none_online_sorted_bounded_lower}, which we will
prove soon, we conclude that $NFD$ is best possible.  It should be
remarked that since $NFD$ is a 1-bounded space algorithm, our bound is
tight for pre-sorted online 1-bounded space algorithms as well as for
pre-sorted online bounded space algorithms.  In other words,
pre-sorted online algorithms do not gain from maintaining a large
number of open bins, as long as the number is bounded by a constant.

Theorem~\ref{thm:none_online_sorted_bounded_lower}
intuitively follows from the fact that 
any pre-sorted online bounded space algorithm can not combine items of different sizes
if the item sequence contains a sufficiently long sequence of
items of the same size.
%%
%% Intuitively, any bounded-space algorithm can be shown to be inherently unable to bin items of different sizes if the input contains a long enough sequence of items of the same size.
%%
The item sequence used to prove
Theorem~\ref{thm:none_online_sorted_bounded_lower} is essentially the
same as that used in the proof of the main theorem in the
paper~\cite{10.1137/0602019}. The reason why the statement of
Theorem~\ref{thm:none_online_sorted_bounded_lower} is nested is to
allow $K$, which is a parameter determined by how close $\delta$ is to
zero, to be chosen independently of $m$.
\begin{theorem}[\cite{10.1137/0602019}]
  %% It holds that $NFD(I) \leq \gamma \cdot OPT(I) - 3$
  It holds that $NFD(I) \leq \gamma \cdot OPT(I) + 3$
  for all item sequences $I$.
  \label{thm:none_online_sorted_bounded_upper}
\end{theorem}
%%
%% The proof of Theorem~\ref{thm:none_online_sorted_bounded_lower}
%% is omitted for space reasons.
%% See Appendix~\ref{subsec:thm:none_online_sorted_bounded_lower} for the
%% proof of Theorem~\ref{thm:none_online_sorted_bounded_lower}.
%%
\begin{theorem}
  For any $\delta$ satisfying $0 < \delta < \gamma$, there exists an
  integer $K \geq 1$ such that: For any pre-sorted online $B$-bounded space
  algorithm $ALG$ and any positive integer $m$, there exists a sorted
  item sequence $I$ such that $OPT(I) = m$ and $ALG(I) > (\gamma -
  \delta) \cdot OPT(I) - B (K - 1)$.
  \label{thm:none_online_sorted_bounded_lower}
\end{theorem}
\begin{proof}
  Fix $\delta$ such that $0 < \delta < \gamma$.  Choose an integer $K
  \geq 1$ such that $\gamma - \delta < \sum_{i=1}^{K}
  \frac{1}{\pi_{i}-1}$.  (\ref{eq:gamma_converge}) and
  $\frac{1}{\pi_{1}-1} = 1$ ensure that this choice is always
  possible.  Fix a max-min $B$-bounded space algorithm $ALG$ and
  a positive integer $m$.  Set $\varepsilon = \frac{1}{K (\pi_{K} -
    1)}$.  Take
  \begin{align*}
    I = \Bigl(
    \overbrace{
      \frac{1}{\pi_{1}} + \varepsilon, \ldots, \frac{1}{\pi_{1}} + \varepsilon
    }^{m},
    \overbrace{
      \frac{1}{\pi_{2}} + \varepsilon, \ldots, \frac{1}{\pi_{2}} + \varepsilon
    }^{m},
    \ldots,
    \overbrace{
      \frac{1}{\pi_{K}} + \varepsilon, \ldots, \frac{1}{\pi_{K}} + \varepsilon
    }^{m}
    \Bigr).
  \end{align*}
  The items in $I$ are located in a non-increasing order.  Without
  loss of generality, assume that $I$ does not change by sorting by
  $ALG$.

  We know $1 - \sum_{i=1}^{K} \frac{1}{\pi_{i}} = \frac{1}{\pi_{K} -
    1}$ by~(\ref{eq:pi_sum}).  Thus, $\sum_{i=1}^{K} \left(
  \frac{1}{\pi_{i}} + \varepsilon \right) = \sum_{i=1}^{K}
  \frac{1}{\pi_{i}} + K \varepsilon = 1$ holds, which means that
  the sum of the sizes of items, one for each size, is 1 and therefore
  $OPT(I) = m$ holds.

  Investigate the assignment by $ALG$ for $I$.  Since $ALG$ is a
  pre-sorted online algorithm, $ALG$ repeatedly packs the head item in the
  currently remaining item sequence into a bin.
  For each $i = 1, 2,
  \ldots, K$, $\left( \frac{1}{\pi_{i}} + \varepsilon \right) \pi_{i}
  = 1 + \varepsilon \pi_{i} > 1$ holds true.  From the monotonically
  increasing property of $\{\pi_{i}\}_{i = 1}^{\infty}$, it follows
  that $\left( \frac{1}{\pi_{i}} + \varepsilon \right) (\pi_{i} - 1) =
  1 - \frac{1}{\pi_{i}} + \varepsilon (\pi_{i} - 1) \leq 1 -
  \frac{\pi_{i} - 1}{\pi_{i+1} - 1} + \frac{\pi_{i} - 1}{\pi_{K+1} -
    1} \leq 1$.  These two facts imply that $(\frac{1}{\pi_{i}} +
  \varepsilon)$-items can be packed into a bin at most $(\pi_{i} -
  1)$.

  Therefore, first, $ALG$ assigns at least $m$ bins to
  $(\frac{1}{\pi_{1}} + \varepsilon)$-items, leaving at most $B$ bins
  open.  Next, for each $i = 2, 3, \ldots, K$, $ALG$ assigns at least
  $(\frac{m}{\pi_{i} - 1} - B)$ bins to $(\frac{1}{\pi_{i}} +
  \varepsilon)$-items, leaving at most $B$ bins open.  In total, we
  bound
  \begin{align*}
    ALG(I) & \geq m + \Bigl(\frac{m}{\pi_{1} - 1} - B \Bigr ) + 
    \Bigl(\frac{m}{\pi_{2} - 1} - B \Bigr ) + \cdots
    + \Bigl( \frac{m}{\pi_{K+1} - 1} - B \Bigr ) \\
    & = \Bigl( \sum_{i=1}^{K} \frac{1}{\pi_{i} - 1} \Bigr ) m - B(K - 1) 
    > (\gamma - \delta) m - B(K - 1),
  \end{align*}
  which proves the theorem.
\end{proof}

%% AAPP
%%
%% 202501060918
%%
%%%%%%%%%%%%%%%%%%%%%%%%%%%%%%%%%%%%%%%%%%%%%%%%%%%%%%%%%%%%%%%%%%%%%%
\section{The $k$-Cardinality Constrained Bin Packing Problem}
\label{sec:k-item}
We consider the \emph{$k$-cardinality constrained} bin packing
problem~\cite{DBLP:journals/siamdm/Epstein06}, which is a variant of
the classic bin packing problem with an additional constraint that
each bin can be assigned to at most $k$ items.  Throughout this
section, for $j = 1, 2, \ldots, k-1$, we call an item of size
$(\frac{1}{j+1}, \frac{1}{j}]$ a \emph{class~$j$} item, and we call an
  item of size $(0, \frac{1}{k}]$ a \emph{class~$k$} item.

An algorithm is called a \emph{$k$-cardinality constrained} algorithm
if for any item sequence, it outputs an assignment that assigns each
bin to at most $k$ items.  We denote by $OPT_{k}(I)$ the minimum
number of bins that are assigned to all items in $I$ while ensuring
that each bin is assigned to at most $k$ items.  We also denote by
$OPT_{k}$ an algorithm that for a given $I$, outputs such an
assignment that the number of non-empty bins is $OPT_{k}(I)$ and each
bin is assigned to at most $k$ items.

The definition of the asymptotic competitive ratio of an algorithm
$ALG$, denoted by $R_{ALG}$, is given by~(\ref{eq:k_acr}).  Similarly
to Section~\ref{sec:none}, bounds on the asymptotic competitive ratio
are obtained from the coefficients of $OPT_{k}(I)$ in the inequalities
in the statements of the theorems.
\subsection{A Max-Min 1-Bounded Space Algorithm}
\label{subsec:k-item_maxmin_unit_algorithm}
%%
%% MM_k
%%
\begin{algorithm}
  \KwIn{An item sequence $I \in (0, 1]^{n}$}
  \KwOut{An assignment $f: \{1, 2, \ldots, n\} \to \mathbb{N}$}
  sort $I$ in a nonincreasing order and obtain $(a_{1}, a_{2}, \ldots, a_{n})$\;
  $h \leftarrow 1$, $t \leftarrow n$\;
  $l \leftarrow 1$, $S \leftarrow 0$\tcp*[r]{open a new bin}
  \While{$h \leq t$}{
    \For(\tcp*[f]{iteration for at most k items}){$i=1$ \KwTo $k$}{
      \lIf(\tcp*[f]{no items left}){$h > t$}{\Return{$f$}}
      \uElseIf(\tcp*[f]{first try head}){$S + a_{h} \leq 1$}{
        $f(h) \leftarrow l$, $S \leftarrow S + a_{h}$\;
        $h \leftarrow h + 1$\;}
      \uElseIf(\tcp*[f]{then try tail}){$S + a_{t} \leq 1$}{
        $f(t) \leftarrow l$, $S \leftarrow S + a_{t}$\;
        $t \leftarrow t - 1$\;}
      \lElse(\tcp*[f]{both head and tail are too large}){break}
    }
    $l \leftarrow l + 1$, $S \leftarrow 0$\tcp*[r]{
      close the bin and open a new bin}
  }
  \Return{$f$}\;
  \caption{$MM_{k}$.}
  \label{algo:mmk}
\end{algorithm}
In this section we propose a max-min 1-bounded space $k$-cardinality
constrained algorithm $MM_{k}$
(see Algorithm~\ref{algo:mmk} for the pseudocode),
which is simply a modified version of algorithm $MM$ in
Section~\ref{sec:none} that opens a new bin every time it has packed
$k$ items into a bin.  Similarly to algorithm $MM$, if algorithm
$MM_{k}$ assigns a bin to the head (respectively, tail) item in the
currently remaining item sequence, then in the resulting assignment we
call that item a \emph{head} (respectively, \emph{tail}) item.

Lemma~\ref{lem:mm_2_subset} states a significant property of the
assignment by $MM_{2}$, which is useful for the analysis not only for
$k = 2$ but also for $k \geq 3$.
%%
%% The proof is omitted due to space limitations.
%% See Appendix~\ref{subsec:lem:mm_2_subset} for the proof.
%%
\begin{lemma}
  Suppose that $I = (a_{1}, a_{2}, \ldots, a_{n}) \in (0, 1]^{n}$ is a
    sorted item sequence such that the assignment $f$ by $MM_{2}$ for
    $I$ has at least one bin with a sole class~1 item.  Let $p$ be the
    maximum index of such sole class~1 items in $f$. Let $U$ be the
    set of indices of items that are packed with items $1, 2, \ldots,
    p$ in $f$, which may be empty.  Then, for an arbitrary assignment
    without cardinality constraints $g$ for $I$, denoting by $W$ the
    set of indices of items that are packed with items $1, 2, \ldots,
    p$ in $g$, it holds that $W \subseteq U$.
  \label{lem:mm_2_subset}
\end{lemma}
\begin{proof}
  Since $I$ is sorted, without loss of generality, assume that $I$
  does not change by sorting by $MM_{2}$.  Since
  items $1, 2, \ldots, p$ have size at least $a_{p}$,
  all items in $W$ are of size at most $1 -
  a_{p}$. Therefore, to prove the theorem, it suffices to show that
  $U$ is the set of all items in $I$ that are of size at most $1 -
  a_{p}$.

  Immediately after $MM_{2}$ has packed item $p$, there should be no
  items of size at most $1 - a_{p}$ in the remaining item sequence.
  (Otherwise, some of them would have been packed with item $p$, which
  contradicts the definition of $p$.)  On the other hand, $MM_{2}$
  assigns each of the bins earlier than the bin containing item $p$ to one
  of items $1, 2, \ldots, p-1$.  This means that if $I$ has items of
  size at most $1 - a_{p}$, such items all belong to $U$.
\end{proof}
%%
%% AAPP
%%
\subsubsection{A Tight Bound for $k = 2$}
\label{subsubsec:2-item_maxmin_unit_upper}
Tinhofer showed that if $FFD$, which is a pre-sorted online unbounded
space algorithm, is run for the $2$-cardinality constrained bin
packing problem, it always outputs an assignment with the minimum
number of non-empty bins~\cite{Tinhofer1995BinpackingAM}.  That is,
the asymptotic approximation ratio is 1.  The following theorem
implies that if we are allowed to employ a max-min algorithm, then a
single open bin is enough to enjoy the optimal performance.
%%
%% The proof is omitted due to lack of space.
%% See Appendix~\ref{subsec:thm:tsuchiya} for the proof.
%%
\begin{theorem}
  It holds that $MM_{2}(I) = OPT_{2}(I)$ for all item sequences $I$.
  \label{thm:tsuchiya}
\end{theorem}
\begin{proof}
  Without loss of generality, let $I = (a_{1}, a_{2}, \ldots, a_{n})
  \in (0, 1]^{n}$ be a sorted item sequence and assume that $I$ does
  not change by sorting by algorithm $MM_{2}$.
  Let $f$ be the assignment by
  $MM_{2}$ for $I$.  If $f$ does not have a bin containing a sole
  class~1 item, all the bins in $f$, except the latest bin, must each
  contain two items, which means that $|f(I)|$ is minimum and
  therefore the proof is complete.
    
  From now and on, assume that $f$ has at least one bin with a sole
  class~1 item.  As in Lemma ~\ref{lem:mm_2_subset}, let $p$ be the
  maximum index of such sole class~1 items in $f$, and $U$ be the set of
  indices of items that are packed with items $1, 2, \ldots, p$ in $f$,
  which may be empty.  On the other hand, let $g$ be an arbitrary
  assignment without cardinality constraints for $I$,
  and $W$ be the set of indices of items that are
  packed with items $1, 2, \ldots, p$ in $g$. To prove the theorem, it
  suffices to show that $|g(I)| \geq |f(I)|$ holds.

  Let us evaluate $|f(I)|$. Since item $p$ is the sole class~1 item to
  which algorithm $MM_{2}$ lastly assigns a bin, $MM_{2}$ then
  continues to pack each class~1 item with another item, or a pair of
  two class~2 items into a bin.  That is to say, in $f$, items $1, 2,
  \ldots, p$ are sole or contained with another item, and other items form
  as many pairs as possible.  Hence, it holds that $|f(I)| = p +
  \lceil \frac{n - p - |U|}{2} \rceil$.

  On the other hand, $|g(I)|$ is bounded from below as follows. Items
  $1, 2, \ldots, p$ are all class~1 items and therefore at most one of
  them can fit in a bin.  In $g$, items in $W$ are packed with items
  $1, 2, \ldots, p$.  From these observations, we know that $|g(I)|$
  is minimum if all the items in $I$, other than items $1, 2, \ldots, p$
  and $W$, are packed two per bin.  We thus obtain $|g(I)| \geq p +
  \lceil \frac{n - p - |W|}{2} \rceil$.  Here, $W \subseteq U$ follows
  from Lemma~\ref{lem:mm_2_subset} and therefore $|W| \leq |U|$ holds.
  Hence, we get $p + \lceil \frac{n - p - |W|}{2} \rceil \geq p +
  \lceil \frac{n - p - |U|}{2} \rceil = |f(I)|$.
\end{proof}
%%
%% AAPP
%%
\subsubsection{An Upper Bound for $k \geq 3$}
\label{subsubsec:3ormore-item_maxmin_unit_upper}
Recall the sequence $\{\pi_{i}\}_{i = 1}^{\infty}$ defined by $\pi_{1}
= 2$ and $\pi_{i+1} = \pi_{i}(\pi_{i} - 1) + 1$, and the sequence
$\{\lambda_{j}\}_{j = 1}^{\infty}$ defined by $\lambda_{j} =
\sum_{i=1}^{j} \max \{\frac{1}{\pi_{i}-1}, \frac{1}{j}\}$ for $j \geq
1$ in Section~\ref{subsec:ourcontribution}.
For example, the values of $\{\lambda_{j}\}_{j = 1}^{\infty}$
are: $\lambda_{1} = 1$, $\lambda_{2} = \frac{3}{2} = 1.5$,
$\lambda_{3} = \frac{11}{6} \approx 1.83$, $\lambda_{4} = 2$,
and $\lambda_{5} = \frac{21}{10} = 2.1$.

For the $k$-cardinality constrained bin packing problem with $k \geq
3$, it is not so difficult to obtain $R_{MM_{k}} \leq \lambda_{k}$ by
a standard analysis using weights.  In what follows, we conduct a more
precise analysis to reveal that $R_{MM_{k}} \leq \lambda_{k} -
\frac{1}{k}$, breaking the lower bound barrier of $\lambda_{k}$ for
pre-sorted online bounded space
algorithms~\cite{DBLP:journals/siamdm/Epstein06}.  For example, the
values of $\lambda_{k} - \frac{1}{k}$ are: $\frac{3}{2} = 1.5$ for $k
= 3$, $\frac{7}{4} = 1.75$ for $k = 4$, and $\frac{19}{10} = 1.9$ for
$k = 5$.

%% The main idea of our analysis is to dare to run $MM_{2}$ as well for
The main idea of our analysis is to attempt to run $MM_{2}$ as well for
the item sequence in order to extract a subsequence of items to which
an analysis using ``discounted'' weights can be applied.
Lemma~\ref{lem:mm_2_atsumi_9}, proven with the help of
Lemma~\ref{lem:mm_2_subset}, enables us to extract such a subsequence.
%%
%% The proof of Lemma~\ref{lem:mm_2_atsumi_9} is omitted for space reasons.
%% See Appendix~\ref{subsec:lem:mm_2_atsumi_9} for the proof.
%%
\begin{lemma}
  Let $I$ be a sorted item sequence, and $I_{1}$ be a subsequence of
  $I$ consisting of the sole class~1 items in the assignment $f$ by
  $MM_{2}$ for $I$.  Then, for an arbitrary assignment without
  cardinality constraints $g$ for $I$, there exists an assignment without cardinality constraints $g'$
  such that $|g(I)| = |g'(I)|$ and
  every item in $I_{1}$ is a sole item in $g'$.
  \label{lem:mm_2_atsumi_9}
\end{lemma}
\begin{proof}
  If $I_{1}$ is empty or every item in $I_{1}$ is
  a sole item in $g$, then the proof is complete.  From now and on,
  assume that neither is the case.

  Let $W$ and $U$ be the sets of indices of items that are packed
  together with items $1, 2, \ldots, p$ in $g$ and $f$, respectively.
  By Lemma~\ref{lem:mm_2_subset}, $W \subseteq U$ holds.  Therefore,
  we can repack each item in $W$ into the bin that is assigned to it
  in $f$, without opening new bins.  More specifically, we define a
  new assignment $g'$ as follows: For each $i = 1, 2, \ldots, n$,
  $g'(i) := f(i)$ if $i \in W$ and $g'(i) := g(i)$ otherwise.  Here,
  we do not open any new bins, therefore $|g(I)| = |g'(I)|$.  Since to
  each item in $W$, $g'$ assigns a bin that does not contain items in
  $I_{1}$, every item $I_{1}$ is a sole item in $g'$.
\end{proof}
%%
%% AAPP
%%
We employ the following weight function which assigns a weight to each
item according to the class.  The weight of $1$ minus $\frac{1}{k}$
for a class~1 item will contribute to the performance evaluation
smaller by $\frac{1}{k}$.
%%
%% OPT Epstein style 202409040957
%%
\begin{alignat}{2}
  w_{2}(x) & =
  \begin{cases}
    1 - \frac{1}{k}, & x \in (\frac{1}{2}, 1];\\
    \frac{1}{j}, & x \in (\frac{1}{j + 1}, \frac{1}{j}]
    \text{ for } j = 2, \ldots, k - 1;\\
    \frac{1}{k}, & x \in (0, \frac{1}{k}].
  \end{cases}
  \label{eq:mm_k_cr_weight}
\end{alignat}
Given $I = (a_{1}, a_{2}, \ldots, a_{n})$, define $w_{2}(I) =
\sum_{i=1}^{n} w_{2}(a_{i})$.  The next theorem bounds from above the
sum of the weights of items that fit in a single bin.  The proof is
given by modifying the proof of the upper bound part of Theorem~3
in~\cite{DBLP:journals/siamdm/Epstein06}. 
%%
%% For lack of space, we omit the proof of Theorem~\ref{lem:mm_k_cr_opt_max_weight}.
%% See
%% Appendix~\ref{subsec:lem:mm_k_cr_opt_max_weight} for the proof.
%%
%%
\begin{lemma}
  For any $k \geq 3$ and any $I = (a_{1}, a_{2}, \ldots, a_{n}) \in
  (0, 1]^{n}$ such that $n \leq k$ and $\sum_{i=1}^{n} a_{i} \leq 1$,
    it holds that $w_{2}(I) \leq \lambda_{k} - \frac{1}{k}$.
  \label{lem:mm_k_cr_opt_max_weight}
\end{lemma}
\begin{proof}
  We begin by giving an item sequence whose sum of the weights
  by $w_{2}$ is exactly $\lambda_{k} - \frac{1}{k}$: Such an item sequence 
  is obtained by taking the size $b_{i} = \frac{1}{\pi_{i}} + \varepsilon$
  for $i = 1, 2, \ldots, k$, where $\varepsilon$ is a constant such that
  $0 < \varepsilon \leq \frac{1}{k(\pi_{k+1}-1)}$.  By 
  (\ref{eq:pi_sum}), we have $\sum_{i=1}^{k} b_{i} \leq \sum_{i=1}^{k}
  \frac{1}{\pi_{k}} + k \cdot \frac{1}{k(\pi_{k+1}-1)} = 1$, 
  %% which implies that these items can be packed into a single bin.
  which implies that these items can fit in a bin.
  Since $b_{i} \in (\frac{1}{\pi_{i}}, \frac{1}{\pi_{i} - 1}]$, 
    we calculate $\sum_{i=1}^{k} w_{2}(b_{i}) = 1 - \frac{1}{k} +
    \sum_{i=2}^{k} \max \{\frac{1}{\pi_{i}-1}, \frac{1}{k}\} =
    \lambda_{k} - \frac{1}{k}$.  We will use this $(b_{1}, b_{2},
    \ldots, b_{k})$ throughout the proof.

  Every item sequence that maximizes the sum of the weights consists
  of exactly $k$ items, since any item sequence of less than $k$ items can
  increase its weight by replacing each item with a slightly smaller
  one without changing the weight, and
  adding an infinitesimal item of weight $\frac{1}{k}$.

  Now, we are going to prove that an arbitrarily fixed $I = (a_{1},
  a_{2}, \ldots, a_{k})$ such that $a_{1} \geq a_{2} \geq \cdots \geq
  a_{k}$ satisfies $w_{2}(I) \leq \sum_{i=1}^{k} w_{2}(b_{i}) (=
  \lambda_{k} - \frac{1}{k})$.  If $a_{i} \in (\frac{1}{\pi_{i}},
  \frac{1}{\pi_{i} - 1}]$ for $i$ such that $\pi_{i} \leq k$, and
    $a_{i} \in (0, \frac{1}{k}]$ for $i$ such that $\pi_{i} > k$, then
      $w_{2}(I) = \lambda_{k} - \frac{1}{k}$, and the proof is
      complete. Otherwise, let $p$ be the first index that does not
      satisfy the above. If $w_{2}(a_{p}) = \frac{1}{k}$, then
      $w_{2}(I) = \sum_{i=1}^{p-1} w_{2}(a_{i}) + \sum_{i=p}^{k}
      w_{2}(a_{i}) = \sum_{i=1}^{p-1} w_{2}(b_{i}) + \sum_{i=p}^{k}
      \frac{1}{k} \leq \sum_{i=1}^{k} w_{2}(b_{i})$, and the proof is
      complete in this case as well.

  In the remainder of the proof, assume that $w_{2}(a_{p}) >
  \frac{1}{k}$. Let $q$ be the smallest index for which $w_{2}(a_{q})
  = \frac{1}{k}$. By $a_{1} \geq a_{2} \geq \cdots \geq a_{k}$, it
  should hold that $a_{k} \leq \frac{1}{k}$, thus such an index must
  exist.

  (I)~The case where $p \geq 2$.  The proof is done as the same as the
  upper bound part of the proof of Theorem~3
  in~\cite{DBLP:journals/siamdm/Epstein06}.  Since $\{a_{i}\}_{j =
    1}^{k}$ is non-increasing and $a_{p} \notin (\frac{1}{\pi_{p}},
  \frac{1}{\pi_{p} - 1}]$, $a_{p} \leq \frac{1}{\pi_{p}}$ and
    $w_{2}(a_{p}) \leq \frac{1}{\pi_{p}} < \frac{1}{\pi_{p} - 1} =
    w_{2}(b_{p})$ hold.

  If $q = p + 1$, then we get $w_{2}(I) = \sum_{i=1}^{p-1}
  w_{2}(b_{i}) + w_{2}(a_{p}) + \sum_{i=q}^{k} \frac{1}{k} <
  \sum_{i=1}^{p-1} w_{2}(b_{i}) + w_{2}(b_{p}) + \sum_{i=q}^{k}
  w_{2}(b_{i}) = \sum_{i=1}^{k} w_{2}(b_{i})$.

 For the case of $q \geq p + 2$, let us focus on the ``efficiency'' of
 weights.  From (\ref{eq:mm_k_cr_weight}), we know that if $x \in
 (0, \frac{1}{\pi_{p}}]$ then $\frac{w_{2}(x)}{x} < \frac{\pi_{p} +
     1}{\pi_{p}}$.  In addition, we evaluate the total size of items
   from $p$ to $k$: $\sum_{i=p}^{k} a_{i} \leq 1 - \sum_{i=1}^{p-1}
   a_{i} < 1 - \sum_{i=1}^{p-1} \frac{1}{\pi_{i}} = \frac{1}{\pi_{p} -
     1}$, where we used $a_{i} \in (\frac{1}{\pi_{i}},
   \frac{1}{\pi_{i} - 1}]$ for $i \leq p - 1$, as well
     as~(\ref{eq:pi_sum}).  Therefore, we get $\sum_{i=p}^{q-1}
     w_{2}(a_{i}) = \sum_{i=p}^{q-1} \frac{w_{2}(a_{i})}{a_{i}} \cdot
     a_{i} < \frac{\pi_{p} + 1}{\pi_{p}} \cdot \frac{1}{\pi_{p} - 1} =
     \frac{\pi_{p} + 1}{\pi_{p}(\pi_{p} - 1)}$.  On the other hand, we
     have $w_{2}(b_{p}) + w_{2}(b_{p+1}) = \frac{1}{\pi_{p} - 1} +
     \frac{1}{\pi_{p+1} - 1} = \frac{1}{\pi_{p} - 1} +
     \frac{1}{\pi_{p}(\pi_{p} - 1)} = \frac{\pi_{p} +
       1}{\pi_{p}(\pi_{p} - 1)}$.  In total, $w_{2}(I) =
     \sum_{i=1}^{p-1} w_{2}(a_{i}) + \sum_{i=p}^{q-1} w_{2}(a_{i}) +
     \sum_{i=q}^{k} w_{2}(a_{i}) \leq \sum_{i=1}^{p-1} w_{2}(b_{i}) +
     (w_{2}(b_{p}) + w_{2}(b_{p+1})) + \sum_{i=q}^{k} \frac{1}{k} \leq
     \sum_{i=1}^{k} w_{2}(b_{i})$.

  (II)~The case where $p = 1$. $w_{2}(a_{1}) > \frac{1}{k}$ and $a_{1}
     \notin (\frac{1}{2}, 1]$ imply that $\frac{1}{k} < a_{1} \leq
       \frac{1}{2}$.  Together with $k \geq 3$, we evaluate
       $w_{2}(a_{1}) \leq \frac{1}{2} < 1 - \frac{1}{k} =
       w_{2}(b_{1})$.

  If $q = 2$, then the proof is complete since $w_{2}(I) =
  w_{2}(a_{1}) + \sum_{i=2}^{k} \frac{1}{k} < w_{2}(b_{1}) +
  \sum_{i=2}^{k} w_{2}(b_{i}) = \sum_{i=1}^{k} w_{2}(b_{i})$.

  For the case of $q \geq 3$, we are going to evaluate $w_{2}(I)$
  differently depending on the value of $k$.  In preparation for this,
  we bound $\sum_{i=2}^{q-1} w_{2}(a_{i})$ from above.  If $q = 3$,
  \begin{equation}
    \sum_{i=2}^{q-1} w_{2}(a_{i}) = w_{2}(a_{2}) \leq \frac{1}{2}
    \label{eq:case_2_efficiency_1}
  \end{equation}
  is easily obtained since $\{a_{i}\}_{j = 1}^{k}$ is non-increasing.
  For the case of $q \geq 4$, let us focus on the ``efficiency'' of
  weights.  Let $u$ be the index such that $a_{1} \in
  (\frac{1}{u}, \frac{1}{u - 1}]$.  (\ref{eq:mm_k_cr_weight}) says
    that if $x \in (0, \frac{1}{u - 1}]$, then $\frac{w_{2}(x)}{x} <
      \frac{u}{u - 1}$.  In addition, we evaluate the total size of
      items from $2$ to $k$: $\sum_{i=2}^{k} a_{i} \leq 1 - a_{1} < 1
      - \frac{1}{u}$.  Therefore,
      \begin{equation}
        \sum_{i=2}^{q-1} w_{2}(a_{i}) =
        \sum_{i=2}^{q-1} \frac{w_{2}(a_{i})}{a_{i}} \cdot a_{i} <
        \frac{u}{u - 1} \Bigl( 1 - \frac{1}{u} \Bigr ) = 1
        \label{eq:case_2_efficiency_2}
      \end{equation}
      Now we are ready.

  (II-a)~The case where $3 \leq k \leq 6$.  Noting that
      $\frac{1}{\pi_{i} - 1} \leq \frac{1}{6} \leq \frac{1}{k}$ for $i
      \geq 3$, we have $\sum_{i=1}^{k} w_{2}(b_{i}) = \left(1 -
      \frac{1}{k}\right) + \frac{1}{2} + (k - 2) \cdot \frac{1}{k} =
      \frac{5}{2} - \frac{3}{k}$.  If $q = 3$, we evaluate $ w_{2}(I)
      =w_{2}(a_{1}) + w_{2}(a_{2}) + (k - 2) \cdot \frac{1}{k} \leq
      \frac{1}{2} + \frac{1}{2} + (k - 2) \cdot \frac{1}{k} =2 -
      \frac{2}{k} =(\frac{5}{2} - \frac{3}{k}) - (\frac{1}{2} -
      \frac{1}{k}) < \sum_{i=1}^{k} w_{2}(b_{i}) $ by
      applying~(\ref{eq:case_2_efficiency_1}).  If $q \geq 4$, we
      obtain $ w_{2}(I) = w_{2}(a_{1}) + \sum_{i=2}^{q-1} w_{2}(a_{i})
      + \sum_{i=q}^{k} w_{2}(a_{i}) < \frac{1}{2} + 1 + (k - q + 1)
      \cdot \frac{1}{k} = \frac{5}{2} - \frac{q - 1}{k} = (\frac{5}{2}
      - \frac{3}{k}) - \frac{q - 4}{k} \leq \sum_{i=1}^{k}
      w_{2}(b_{i}) $ from~(\ref{eq:case_2_efficiency_2}).

  (II-b) The case where $k \geq 7$.  If $q = 3$, we obtain $ w_{2}(I)
      = w_{2}(a_{1}) + w_{2}(a_{2}) + \sum_{i=3}^{k} w_{2}(a_{i}) \leq
      \frac{1}{2} + \frac{1}{2} + \sum_{i=3}^{k} \frac{1}{k} < (1 -
      \frac{1}{k}) + \frac{1}{2} + \sum_{i=3}^{k} \frac{1}{k} =
      w_{2}(b_{1}) + w_{2}(b_{2}) + \sum_{i=3}^{k} \frac{1}{k} <
      \sum_{i=1}^{k} w_{2}(b_{i}) $ by~(\ref{eq:case_2_efficiency_1}).
      For $q \geq 4$, noting that $w_{2}(b_{1}) + w_{2}(b_{2}) +
      w_{2}(b_{3}) = \frac{1}{\pi_{1} - 1} + \frac{1}{\pi_{2} - 1} +
      \frac{1}{\pi_{3} - 1} = \left(1 - \frac{1}{k}\right) +
      \frac{1}{2} + \frac{1}{6} = \frac{5}{3} - \frac{1}{k} >
      \frac{3}{2}$, we evaluate $ w_{2}(I) = w_{2}(a_{1}) +
      \sum_{i=2}^{q-1} w_{2}(a_{i}) + \sum_{i=q}^{k} w_{2}(a_{i}) <
      \frac{1}{2} + 1 + \sum_{i=q}^{k} \frac{1}{k} < w_{2}(b_{1}) +
      w_{2}(b_{2}) + w_{2}(b_{3}) + \sum_{i=q}^{k} \frac{1}{k} \leq
      \sum_{i=1}^{k} w_{2}(b_{i}) $ by
      applying~(\ref{eq:case_2_efficiency_2}).
\end{proof}
%%
%% AAPP
%%
We are now ready to claim that $R_{MM_{k}} \leq \lambda_{k} -
\frac{1}{k}$.
%%
%% The proof of Theorem~\ref{thm:mm_k_cr} is omitted for space reasons.
%% See Appendix~\ref{subsec:thm:mm_k_cr} for the proof of
%% Theorem~\ref{thm:mm_k_cr}.
%%
%%
\begin{theorem}
  For any $k \geq 3$, it holds that $MM_{k}(I) \leq \left(\lambda_{k}
  - \frac{1}{k}\right) \cdot OPT_{k}(I) + k$ for all item sequences
  $I$.
  \label{thm:mm_k_cr}
\end{theorem}
\begin{proof}
  Without loss of generality, let $I = (a_{1}, a_{2}, \ldots, a_{n})
  \in (0, 1]^{n}$ be a sorted item sequence and assume that $I$ does
    not change by sorting by algorithm $MM_{k}$.  Let $f$ and $g$ be
    the assignment for $I$ by $MM_{2}$ and $MM_{k}$, respectively.
    Let $I_{1}$ be the subsequence of $I$ that consists of all sole
    class~1 items in $f$.  ($I_{1}$ may be empty, but this does not
    affect the proof.) Let $I_{2}$ be the item subsequence obtained by
    removing $I_{1}$ from $I$.
    
    (I)~We begin by claiming that $MM_{k}(I) = MM_{k}(I_{1}) +
    MM_{k}(I_{2})$.  It is seen from the pseudocode that in both $f$
    and $g$, all class~1 items are head items and, starting from bin
    1, they are packed in a non-increasing order of size into bins.
    Since algorithm $MM_{k}$ has a looser cardinality constraint than
    $MM_{2}$, an item which is packed with a class~1 item in $g$ is
    packed in the same bin in $f$ or in an earlier bin than that bin.
    Thus, if $MM_{k}$ fails to pack an item with a class~1 item, then
    $MM_{2}$ also fails to pack any item with the same class~1 item.
    That is, each sole class~1 item in $f$ is a sole item in $g$ as
    well.  This observation further implies that the assignment
    by $MM_{2}$ for $I_{2}$ is identical, except for
    the difference in bin indices, to the assignment obtained by
    removing all the bins that contain a sole item from $g$.  As a result,
    we conclude that
    \begin{align*}
      MM_{k}(I) = MM_{k}(I_{1}) + MM_{k}(I_{2}).
    \end{align*}

    (II)~We next evaluate $MM_{k}(I)$. Since $I_{1}$ consists of only
    class~1 items, it is easy to obtain that $MM_{k}(I_{1}) =
    OPT_{k}(I_{1}) = |I_{1}|$.  Then, let us consider $MM_{k}(I_{2})$.
    Let $f_{2}$ and $g_{2}$ be the assignment for $I_{2}$ by $MM_{2}$
    and $MM_{k}$, respectively.  Note that there may be sole class~1
    item in $g_{2}$.  We will now construct an assignment $g_{2}'$
    based on assignment $g_{2}$, which satisfies $g_{2}'(I_{2}) =
    g_{2}(I_{2})$ but has no sole class~1 items.

    First, set the assignment $g_{2}'$ to be exactly the same as the
    assignment $g_{2}$.  Then, repeat the following modification
    procedure while the assignment $g_{2}'$ has a sole class~1 item:
    Find the maximum index $z$ among sole class~1 items in $g_{2}'$.
    Formally, $z = \max \{i | a_{i} \in I_{2},
    |{g_{2}'}^{-1}(g_{2}'(i))| = 1, a_{i} \in (\frac{1}{2}, 1]\}$.
      Then, find the index $y$ such that $f_{2}(y) = f_{2}(z)$ and $y
      \neq z$.  That is, item $y$ is packed together with item $z$ in
      $f_{2}$.  This $y$ must uniquely exist, since there is no sole
      class~1 item in $f_{2}$ by the definition of $I_{2}$ and
      $MM_{2}$ is a $2$-cardinality constrained algorithm.  Given such
      $y$, modify $g_{2}'(y) := f_{2}(y)$. Now item $z$ is no longer a
      sole item in $g_{2}'$ (see Figure~\ref{fig:mm_k_cr}).
      \begin{figure}[t]
        \centering
        \includegraphics[width=0.8\linewidth]{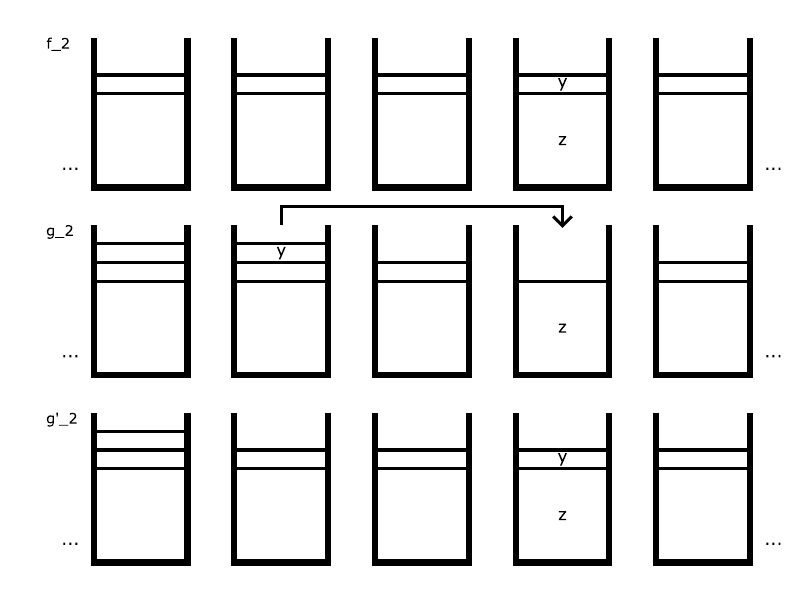}
        \caption{Construction of $g_{2}'$ from $g_{2}$ and $f_{2}$.}
        \label{fig:mm_k_cr}
      \end{figure}
      We below confirm that this procedure does not get stuck.
      Immediately before each modification of $g_{2}'$,
      the bin assigned to item $y$ is earlier than that assigned to item $z$.
      If this were not the case, $MM_{k}$ would have
      packed an item of size at most $a_{y}$ from the tail of the
      remaining item sequence together with item $z$.
      After a modification of $g_{2}'$,
      a class~1 item may become a sole item.
      However, the above observation implies that
      the index of such a new sole class~1 item is smaller than $z$.
      Therefore, this procedure terminates after a finite number of iterations.

      Consider all the bins in $g_{2}'$ other than the latest bin and
      the sum of the weights of the \emph{head} items packed in each
      of these bins.  Observe that algorithm $MM_{k}$ greedily packs
      the head of the currently remaining item sequence into these bins.
      Since there are $k$ classes in total, at most $k - 1$ of these
      bins contain head items of different classes, while the rest
      contain head items of a single class. Furthermore, each of the
      latter bins contains exactly $j$ of class~$j$ items for some $1
      \leq j \leq k$.  If one of the latter bins contains only class~1
      head items, then the sum of the weights of the head items in
      that bin is $(1 - \frac{1}{k}) + \frac{1}{k} = 1$,
      since the assignment $g_{2}'$ was constructed so that no bins
      are assigned only to a single class~1 item.  If one of the
      latter bins contains only class~$j$ head items for $2 \leq j
      \leq k$, then the weights of the head items in that bin is
      $j \cdot \frac{1}{j} = 1$. Consequently, we derive
      \begin{align*}
        w_{2}(I_{2}) \geq MM_{k}(I_{2}) - 1 - (k - 1)
        = MM_{k}(I_{2}) - k.
      \end{align*}
      On the other hand, by Lemma~\ref{lem:mm_k_cr_opt_max_weight},
      the sum of the weights of items that can fit in a bin is
      $\lambda_{k} - \frac{1}{k}$.  Hence, it follows that
      \begin{align*}
        MM_{k}(I_{2}) - k \leq
        w_{2}(I_{2}) \leq \Bigl(\lambda_{k} - \frac{1}{k}\Bigr )
        \cdot OPT_{k}(I_{2}).
      \end{align*}

      (III)~Lemma~\ref{lem:mm_2_atsumi_9} implies that among
      assignments for $I$ with a minimum number of non-empty bins, there exists
      an assignment in which every item in $I_{1}$ is a sole item.
      Therefore, $ OPT_{k}(I) = OPT_{k}(I_{1}) + OPT_{k}(I_{2})$
      holds.  Finally, we obtain
      \begin{align*}
        MM_{k}(I) & = MM_{k}(I_{1}) + MM_{k}(I_{2}) \nonumber \\
        & \leq OPT_{k}(I_{1}) + \Bigl(\lambda_{k} - \frac{1}{k}\Bigr )
        \cdot OPT_{k}(I_{2}) + k
        \\
        & \leq \Bigl(\lambda_{k} - \frac{1}{k}\Bigr ) \cdot
        (OPT_{k}(I_{1}) + OPT_{k}(I_{2})) + k
        \\
        & \leq \Bigl(\lambda_{k} - \frac{1}{k}\Bigr ) \cdot OPT_{k}(I) + k.
      \end{align*}
\end{proof}
%%
%% AAPP
%%
\subsection{A Lower Bound for Max-Min 1-Bounded Space Algorithms}
\label{subsec:k-item_maxmin_unit_lower}
\subsubsection{A Lower Bound for 1-Bounded Space Algorithms}
\label{subsubsec:k-item_maxmin_1-bounded_lower}
Theorem~\ref{thm:3-item_maxmin_unit_lower} below states that $R_{ALG}
\geq \frac{4}{3} (\approx 1.33)$ holds for any max-min 1-bounded space
3-cardinality constrained algorithm $ALG$.
%% As we will mention in
%% Section~\ref{subsec:k-item_maxmin_bounded_lower}, a lower bound for $k
%% \geq 4$ is given as a special case of
%% Theorem~\ref{thm:k-item_maxmin_bounded_lower}.
%%
The proof of Theorem~\ref{thm:3-item_maxmin_unit_lower} is done in
%% almost the same way to the proof of
almost the same way as the proof of
Theorem~\ref{thm:none_maxmin_unit_lower}, except for the cardinality
constraint.
%%
%% The proof is omitted due to space limitations.
%% See Appendix~\ref{subsec:thm:3-item_maxmin_unit_lower} for the
%% proof.
%%
%% 202411271019
%%
\begin{theorem}
  For any max-min 1-bounded space 3-cardinality constrained algorithm
  $ALG$, and any positive integer $m$, there exists a sorted item
  sequence $I$ such that $OPT_{3}(I) = m$ and
  $ALG(I) \geq \frac{4}{3} \cdot OPT_{3}(I) - \frac{4}{3}$.
  \label{thm:3-item_maxmin_unit_lower}
\end{theorem}
\begin{proof}
  We construct the same item sequence in Theorem~\ref{thm:none_maxmin_unit_lower}.
  Fix a max-min 1-bounded space algorithm $ALG$ and a positive integer
  $m$ divisible by 2.  Choose $r$, $\varepsilon$, and $\delta$ such
  that $0 < r < \frac{\sqrt{5} - 1}{2}$,
  $0 < \varepsilon$,
  $0 < \delta$,
  $\varepsilon + \delta < \frac{1}{4} - \frac{1}{5} = \frac{1}{20}$,
  and $\delta r^{m-3} (1 - r - r^{2}) > 2 \varepsilon$.
  (For example, $r = \frac{\sqrt{3} - 1}{2}$,
  $\varepsilon = \frac{9r^{m-3}+10}{80(r^{m-3}+4)}$, and $\delta =
  \frac{13}{40(r^{m-3}+4)}$ satisfy the condition.)  Set $a_{i} =
  \frac{1}{2} + \delta r^{i-1}$ and $b_{i} = \frac{1}{4} - \varepsilon
  - \delta r^{i-1}$ for $i = 1, 2, \ldots, m$, as well as $c =
  \frac{1}{4} + \varepsilon$.  Take
  \begin{align*}
    I = (
    a_{1}, a_{2}, \ldots, a_{m},
    \overbrace{
      c, c, \ldots, c
    }^{m},
    b_{m}, b_{m-1}, \ldots, b_{1}).
  \end{align*}
  The items in $I$ are located in a non-increasing order.  Without
  loss of generality, assume that $I$ does not change by sorting by
  $ALG$.  Note that the order of the indices of $b_{\cdot}$-items is
  the reverse of their order of appearance.  Since $\frac{1}{2} <
  a_{i} < \frac{1}{2} + \delta < 1$, $\frac{1}{4} < c < \frac{1}{2} -
  \frac{1}{5} < \frac{1}{3}$, $\frac{1}{5} < \frac{1}{4} - \varepsilon
  - \delta < b_{i} < \frac{1}{4}$, and $ALG$ is a 3-cardinality
  constrained algorithm, each bin in the assignment by $ALG$ for
  $I$ can contain at most one $a_{\cdot}$-item, at most three
  $c$-items, or at most three $b_{\cdot}$-items.

  It is seen that $a_{i} + c + b_{i} = 1$ for $i = 1, 2, \ldots, m$,
  which means that if the $a_{i}$-item, a $c$-item, and the $b_{i}$-item
  are packed into a bin, the bin is just full. Therefore, $OPT_{3}(I) = m$
  holds.

  Let us evaluate $ALG(I)$.  Since $ALG$ is a max-min algorithm, $ALG$
  repeatedly packs either the head or
  tail item in the currently remaining item sequence into a bin.
  Besides, since $ALG$ is a 1-bounded space algorithm, $ALG$
  packs each item always into the latest bin, not into any earlier bin.
  The difference from the proof of
  Theorem~\ref{thm:none_maxmin_unit_lower} is that now $ALG$ is a
  3-cardinality constrained algorithm, therefore the maximum number of
  $b_{\cdot}$-items that can be packed into a bin is 3.  We denote by
  $D_{i}$ the number of $b_{\cdot}$-items that are packed together
  with $a_{1}, a_{2}, \ldots, a_{i}$ for $i = 1, 2, \ldots, m$. Then,
  (\ref{eq:none_few_small_items}) holds as the same as the proof of
  Theorem~\ref{thm:none_maxmin_unit_lower},

  See the assignment by $ALG$ for $I$.  (If we remove one item from
  the bottom left bin of Figure~\ref{fig:none_maxmin_unit_lower} and
  replace the expression under the bin with ``$\geq \frac{m - (p +
    1)}{3}$'', the figure will represent the assignment in this
  proof.)
  
  (i)~The case where $ALG$ packs the $a_{m}$-item before
  the $b_{m}$-item.  As in the proof of
  Theorem~\ref{thm:none_maxmin_unit_lower}, we obtain
  \begin{align*}
    ALG(I) \geq m + \frac{m}{3} = \frac{4}{3} m.
  \end{align*}

  (ii)~If $ALG$ packs the $b_{m}$-item before the $a_{m}$-item. Let $p$ be
  the maximum index of the $a_{\cdot}$-items that are contained in the
  bin that contains the $b_{m}$-item or earlier bins than it.  If there is
  no such $a_{\cdot}$-items, let $p = 0$.
  By~(\ref{eq:none_few_small_items}), there are at most $(p + 1)$ of
  $b_{\cdot}$-items packed together with $a_1, a_2, \ldots, a_p$.
  Thus, in order to pack the remaining $b_{\cdot}$-items, at least
  $\frac{m - (p + 1)}{3}$ bins are needed, which is the only
  difference from the proof of
  Theorem~\ref{thm:none_maxmin_unit_lower}.

  On the other hand, when
  the $b_{m}$-item has just been packed, $(m - p)$ of $a_{\cdot}$-items and $m$ of
  $c$-items are all still unprocessed. $(m - p)$ of $a_{\cdot}$-items
  require $(m - p)$ bins. Since one $a_{\cdot}$-item and at most
  one $c$-item can be in the same bin, at most $(m - p)$ of $m$
  $c$-items are packed together with the $a_{\cdot}$-item. To pack the
  remaining $c$-items, $\frac{m - (m - p)}{3} = \frac{p}{3}$ or more bins
  should be required. Therefore, we get
  \begin{align*}
    ALG(I) \geq \frac{m - (p + 1)}{3} + p + (m - p) + \frac{p}{3}
    =    \frac{4}{3} m - \frac{4}{3}.
  \end{align*}

  The theorem is proved since $\frac{4}{3} m > \frac{4}{3} m - \frac{4}{3}$.
\end{proof}
%%
%% AAPP
%%
%% \subsection{A Lower Bound for Max-Min Bounded Space Algorithms}
%% \label{subsec:k-item_maxmin_bounded_lower}
\subsubsection{A Lower Bound for $B$-Bounded Space Algorithms}
\label{subsubsec:k-item_maxmin_b_bounded_lower}
%%
%% 202411201112
%%
For $k = 3$, we obtain from
Theorem~\ref{thm:none_maxmin_unbounded_lower} that $R_{ALG} \geq
\frac{16}{15} (\approx 1.07)$ holds for any max-min unbounded space
3-cardinality constrained algorithm $ALG$, since none of the bins
appearing in the proof of
Theorem~\ref{thm:none_maxmin_unbounded_lower} contains 4 or more
items.

The following theorem states that for each $k \geq 4$, $R_{ALG} \geq
\max \{\frac{5}{3} - \frac{1}{k-1}, \frac{3}{2}\}$ holds for any
max-min bounded space $k$-cardinality constrained algorithm $ALG$.
This is the best known lower bound also for 1-bounded space
algorithms.  The proof of
Theorem~\ref{thm:k-item_maxmin_bounded_lower} is similar to that of
Theorem~\ref{thm:none_maxmin_bounded_lower}.
%%
%% The proof is omitted due to lack of space.
%% See
%% Appendix~\ref{subsec:thm:k-item_maxmin_bounded_lower} for the proof.
%%
\begin{theorem}
  For any max-min $B$-bounded space $k$-cardinality constrained
  algorithm $ALG$, and any positive integer $m$ that is at least $(k -
  1)B$ and divisible by 3 and $k - 1$, there exists a sorted item
  sequence $I$ such that $OPT_{k}(I) = m$ and: (A)~$ALG(I) \geq
  \left(\frac{5}{3} - \frac{1}{k-1}\right) \cdot OPT_{k}(I) -
  \frac{1}{6} B(k + 5)$ for each $4 \leq k \leq 6$, (B)~$ALG(I) \geq
  \frac{3}{2} \cdot OPT_{k}(I) - 2 B$ for $k = 7$, and (C)~$ALG(I)
  \geq \frac{3}{2} \cdot OPT_{k}(I) - \frac{1}{2} B$ for each $k \geq
  8$.
  \label{thm:k-item_maxmin_bounded_lower}
\end{theorem}
\begin{proof}
  Fix a max-min $B$-bounded space $k$-cardinality constrained
  algorithm $ALG$ and a positive
  integer $m$ that is at least $(k - 1)B$
  and divisible by 3 and $k - 1$.  Choose $\varepsilon$ such that $0 <
  \varepsilon < \frac{1}{3} \cdot (\frac{1}{6} - \frac{1}{7}) =
  \frac{1}{126}$.  Take
  \begin{align*}
    I = \Bigl(
    \overbrace{
      \frac{1}{2} + \varepsilon, \ldots, \underline{\frac{1}{2} + \varepsilon}
    }^{m},
    \overbrace{
      \frac{1}{3} + \varepsilon, \ldots, \frac{1}{3} + \varepsilon
    }^{m},
    \overbrace{
      \underline{\frac{1}{6} - 3 \varepsilon},
      \ldots, \frac{1}{6} - 3 \varepsilon
    }^{m},
    \overbrace{
      \underline{\frac{\varepsilon}{k - 3}},
      \ldots, \frac{\varepsilon}{k - 3}
    }^{(k - 3) m}
    \Bigr).
  \end{align*}
  Here, the last $(\frac{1}{2} + \varepsilon)$-item, the first
  $(\frac{1}{6} - 3 \varepsilon)$-item, and the first
  $(\frac{\varepsilon}{k - 3})$-item are important, so we underlined
  these items.  The items in $I$ are located in a non-increasing
  order.  Without loss of generality, assume that $I$ does not change
  by sorting by $ALG$.

  Since the sum of the item sizes of a $(\frac{1}{2} +
  \varepsilon)$-item, a $(\frac{1}{3} + \varepsilon)$-item, a
  $(\frac{1}{6} - 3 \varepsilon)$-item, and $(k - 3)$ of
  $(\frac{\varepsilon}{k - 3})$-items is exactly one, $OPT(I) = m$
  holds.

  \begin{figure}[t]
    \centering
    \includegraphics[width=1.0\linewidth]{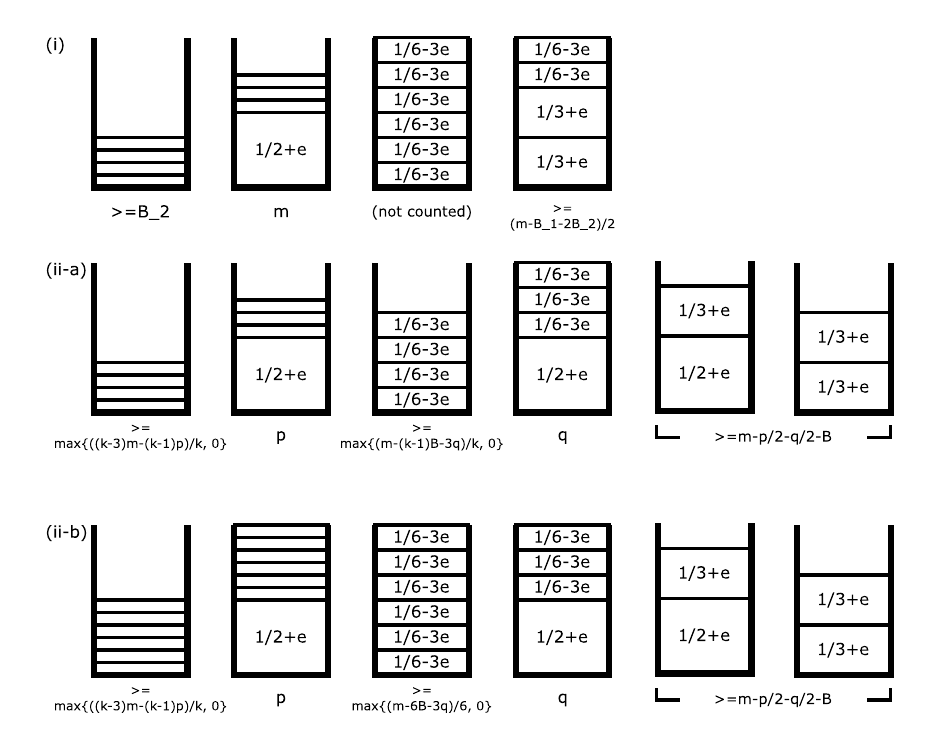}
    \caption{The assignment by $ALG$ for $I$.
    The items without labels are of size $\frac{\varepsilon}{k - 3}$.}
    \label{fig:k-item_maxmin_bounded_lower}
  \end{figure}
  Let us evaluate $ALG(I)$ (see
  Figure~\ref{fig:k-item_maxmin_bounded_lower}).  Since $ALG$ is a
  max-min algorithm, $ALG$ repeatedly packs either the head or tail
  item in the currently remaining item sequence into a bin.

  (i)~The case where $ALG$ packs the last $(\frac{1}{2} +
  \varepsilon)$-item before the first $(\frac{1}{6} - 3
  \varepsilon)$-item.
  Focus on the moment immediately after $ALG$ has packed the last
  $(\frac{1}{2} + \varepsilon)$-item.
  Suppose that 
  there are $B_{1}$ open bins that contain a $(\frac{1}{2} +
  \varepsilon)$-item and $B_{2}$ open bins
  that do not contain a $(\frac{1}{2} +
  \varepsilon)$-item,
  where $B_{1} + B_{2} \leq B$.
  Then, the number of bins assigned
  so far is at least $m + B_{2}$.

  Next, consider all the items currently unprocessed at this moment.
  Ignoring the
  $(\frac{1}{6} - 3 \varepsilon)$-items, just see how the $m$ of
  $(\frac{1}{3} + \varepsilon)$-items are going to be packed.
  While each of the open $B_{1}$ bins
  will be assigned to at most one $(\frac{1}{3} +
  \varepsilon)$-item, each of the open $B_{2}$ bins will be assigned to
  at most two $(\frac{1}{3} + \varepsilon)$-items.
  Besides, the
  remaining $(\frac{1}{3} + \varepsilon)$-items will be packed by
  opening new bins, each with at most two items. This is always
  valid, since $m \geq (k - 1)B \geq 2 B$
  implies $m - B_{1} -  2 B_{2} \geq m - 2 B \geq 0$.
  Therefore, we have
  \begin{align*}
    ALG(I) \geq m + B_{2} + \frac{m - B_{1} - 2 B_{2}}{2} =
    \frac{3}{2} m - \frac{1}{2} B_{1} \geq
    \frac{3}{2} m - \frac{1}{2} B.
  \end{align*}

  (ii)~The case where $ALG$ packs the first $(\frac{1}{6} - 3
  \varepsilon)$-items before the last $(\frac{1}{2} +
  \varepsilon)$-items.

  (ii-a)~The case of $4 \leq k \leq 6$. Let $p$ be the number of
  $(\frac{1}{2} + \varepsilon)$-items which has been packed before the
  first $(\frac{\varepsilon}{k - 3})$-item gets packed. It is seen that
  there can be at most one $(\frac{1}{2} + \varepsilon)$-item in a bin
  and such a bin can additionally contain at most $(k-1)$ of
  $(\frac{\varepsilon}{k - 3})$-items. Then, the minimum number of
  bins assigned before the first $(\frac{\varepsilon}{k - 3})$-item was
  packed is at least $\max\left\{\frac{(k-3)m - (k-1)p}{k}, 0\right\}
  + p$. At most $B$ of these bins remain open, each of which can
  additionally contain at most $(k-1)$ items.

  Next, consider the period after the first $(\frac{\varepsilon}{k -
    3})$-item got packed and before the first $(\frac{1}{6} - 3
  \varepsilon)$-item gets packed. Let $q$ be the number of
%%   $(\frac{1}{2} + \varepsilon)$-items which gets packed during this
  $(\frac{1}{2} + \varepsilon)$-items which get packed during this
  period. Since at most $3$ of $(\frac{1}{6} - 3 \varepsilon)$-items
  are additionally packed into each bin that contains one
  $(\frac{1}{2} + \varepsilon)$-item, the number of bins assigned during
  this period is at least $\max\left\{\frac{m - (k - 1)B - 3q}{k},
  0\right\} + q$. At most $B$ bins of those that have been assigned until the
  end of this period remain open, each of which can additionally
  contain at most 2 of the remaining items.

  At the end of this period, there are $(m - p - q)$ of $(\frac{1}{2}
  + \varepsilon)$-items and $m$ of $(\frac{1}{3} + \varepsilon)$-items
  left.  These items can be packed at most two per bin.  Therefore,
  the number of bins required to pack these items, including the open
  bins mentioned above, is at least $\frac{(m - p - q) + m - 2 B}{2} =
  m - \frac{p}{2} - \frac{q}{2} - B$.

  From these observations, we obtain that $ALG(I)$ is at least
  \begin{multline*}
    \max\Bigl\{\frac{(k-3)m - (k-1)p}{k}, 0\Bigr\} + p +
    \max\Bigl\{\frac{m - (k - 1)B - 3q}{k}, 0\Bigr\} + q +
    m - \frac{p}{2} - \frac{q}{2} - B = \\
    \max\Bigl\{\frac{(k-3)m - (k-1)p}{k} + \frac{p}{2},
    \frac{p}{2}\Bigr\} +
    \max\Bigl\{\frac{m - (k - 1)B - 3q}{k} + \frac{q}{2},
    \frac{q}{2}\Bigr\} +
    m - B.
  \end{multline*}
  The first max operator takes the minimum value
  $\frac{(k-3)m}{2(k-1)}$ when $p = \frac{(k-3)m}{k-1}$, while the
  second max operator takes the minimum value $\frac{1}{6}(m - (k -
  1)B)$ when $q = \frac{1}{3}(m - (k - 1)B)$. These are always valid
  due to the assumption that $m$ is divisible by $3$ and $k-1$, $0
  \leq \frac{1}{3}(m - (k - 1)B) = q$ implied by the assumption $m
  \geq (k - 1)B$, and $m - p- q = \frac{(7-k)m}{3(k-1)} +
  \frac{1}{3}(k - 1)B > 0$. Therefore, it holds that
  \begin{align*}
    ALG(I) \geq
    \frac{(k-3)m}{2(k-1)} +
    \frac{1}{6}(m - (k - 1)B) +
    m - B =
    \Bigl(\frac{5}{3} - \frac{1}{k-1}\Bigr) \cdot m - \frac{1}{6} B(k + 5).
  \end{align*}

  (ii-b)~The case of $k \geq 7$. We evaluate the number of bins in the
  same way as~(ii-a), noting that up to 6 of $(\frac{1}{6} - 3
  \varepsilon)$-items can be contained by one bin. We get
  \begin{align}
    & ALG(I) \nonumber \\
    & \geq \max\Bigl\{\frac{(k-3)m - (k-1)p}{k}, 0\Bigr\} + p +
    \max\Bigl\{\frac{m - 6 B - 3q}{6}, 0\Bigr\} + q +
    m - \frac{p}{2} - \frac{q}{2} - B \nonumber \\
    & \geq \max\Bigl\{\frac{(k-3)m - (k-1)p}{k}, 0\Bigr\} + p +
    \frac{m - 6 B - 3q}{6} + q +
    m - \frac{p}{2} - \frac{q}{2} - B \nonumber \\
    & = \max\Bigl\{\frac{(k-3)m - (k-1)p}{k} + \frac{p}{2},
    \frac{p}{2}\Bigr\} + \frac{7}{6} m - 2 B
    \label{eq:k-item_maxmin_unbounded_lower_max} \\
    & \geq \frac{(k-3)m}{2(k-1)} + \frac{7}{6} m - 2 B
    \nonumber \\
    & = \Bigl(\frac{5}{3} - \frac{1}{k-1} \Bigr) m - 2 B. \nonumber
  \end{align}
  The max operator in~(\ref{eq:k-item_maxmin_unbounded_lower_max})
  takes the minimum value $\frac{(k-3)m}{2(k-1)}$ when $p =
  \frac{(k-3)m}{k-1}$, that is, when both operands are equal. This is
  always valid since we have assumed that $m$ is divisible by $k-1$.

  Finally, compare the lower bound values for each $k$. When $4 \leq k
  \leq 6$, $\left(\frac{5}{3} - \frac{1}{k-1}\right) \cdot m -
  \frac{1}{6} B(k + 5) < \frac{3}{2} m - \frac{1}{2} B$ holds. Thus,
  the lower bound value of~(ii-a) is the smallest, which leads us to
  the claim~(A).  When $k = 7$, $\frac{5}{3} - \frac{1}{k-1} =
  \frac{3}{2}$ holds and the lower bound value of~(ii-b) is the
  smallest, which shows the claim~(B).  When $k \geq 8$, the lower
  bound value of~(i) is the smallest.  That is, the claim~(C) holds
  true.
\end{proof}
%%
%% AAPP
%%
\subsection{A Tight Bound for Pre-Sorted Online 1-Bounded Space Algorithms}
\label{subsec:k-item_online_sorted_unit}
%%
%% 202412250907
%%
We denote by $NFD_{k}$ a $k$-cardinality constrained version of
algorithm $NFD$ (see
Section~\ref{subsec:none_online_sorted_bounded_tight}) that opens a
new bin every time it has packed $k$ items into a bin.  The following
Theorems~\ref{thm:k-item_online_sorted_unit_nextfit_k}
and~\ref{thm:k-item_online_bounded_lower}~\cite{DBLP:journals/siamdm/Epstein06}
state that $NFD_{k}$ is best possible among online bounded space
$k$-cardinality constrained algorithms, achieving $R_{NFD_{k}} =
\lambda_{k}$, though $NFD_{k}$ is a 1-bounded space algorithm.

Compare with the fact that algorithm $CCH_{k}$ (Cardinality
Constrained Harmonic) is best possible among online bounded space
$k$-cardinality constrained algorithms, achieving the same asymptotic
approximation ratio $R_{CCH_{k}} =
\lambda_{k}$~\cite{DBLP:journals/siamdm/Epstein06}.  Our result
explains a trade-off between algorithm classes: if pre-sorting is
allowed, then a single open bin is enough for the same performance.

To analyze the performance of $NFD_{k}$, we use the following weight
function defined in the paper \cite{DBLP:journals/siamdm/Epstein06}:
\begin{alignat*}{2}
  w_{3}(x) & =
  \begin{cases}
    \frac{1}{j}, & x \in (\frac{1}{j + 1}, \frac{1}{j}]
    \text{ for } j = 1, \ldots, k - 1;\\
    \frac{1}{k}, & x \in (0, \frac{1}{k}].
  \end{cases}
\end{alignat*}
Given $I = (a_{1}, a_{2}, \ldots, a_{n})$, define $w_{3}(I) =
\sum_{i=1}^{n} w_{3}(a_{i})$.  The sum of the weights of items that
fit in a single bin is known to be bounded from above:
\begin{lemma}[\cite{DBLP:journals/siamdm/Epstein06}]
  For any $k \geq 2$ and any $I = (a_{1}, a_{2}, \ldots, a_{n}) \in (0, 1]^{n}$
    such
  that $n \leq k$ and $\sum_{i=1}^{n} a_{i} \leq 1$, it holds that
  $w_{3}(I) \leq \lambda_{k}$.
  \label{lem:nextfit_k_cr_opt_max_weight}
\end{lemma}
The theorem below proves $R_{NFD_{k}} \leq \lambda_{k}$.
The proof is done similarly to the analysis of $NFD$ in~\cite{10.1137/0602019},
except for the cardinality constraint.
%%
%% The proof is omitted due to space limitations.
%% See Appendix~\ref{subsec:thm:k-item_online_sorted_unit_nextfit_k} for the
%% proof.
%%
\begin{theorem}
  For any $k \geq 2$, it holds that $NFD_{k}(I) \leq \lambda_{k} \cdot
  OPT_{k}(I) + k$ for all item sequences $I$.
  \label{thm:k-item_online_sorted_unit_nextfit_k}
\end{theorem}
\begin{proof}
  Without loss of generality, let $I = (a_{1}, a_{2}, \ldots, a_{n})
  \in (0, 1]^{n}$ be a sorted item sequence and assume that $I$ does
  not change by sorting by algorithm $NFD_{k}$.

  $NFD_{k}$ greedily packs the head of the currently remaining item
  sequence into the current open bin.  Observe the assignment by
  algorithm $NFD_{k}$ for $I$.  There are at most $(k - 1)$ bins that
  contain items of multiple classes, since there are $k$ classes.
  Except for these bins and the latest bin, all other bins each
  contain exactly $j$ of class~$j$ items for some $1 \leq j \leq k$.
  The sum of weights of the items in such a bin is $j \cdot
  \frac{1}{j} = 1$.  Thus, $w_{3}(I) \geq NFD_{k}(I) - (k - 1) - 1 =
  NFD_{k}(I) - k$ is true.

  On the other hand, if we consider an arbitrary packing that fits
  in a bin, by Lemma~\ref{lem:nextfit_k_cr_opt_max_weight}, the
  sum of the weights is at most
  $\lambda_{k}$. Therefore, $w_{3}(I) \leq \lambda_{k} \cdot
  OPT_{k}(I)$ holds.  Therefore, we have $NFD_{k}(I) \leq \lambda_{k}
  \cdot OPT_{k}(I) + k$.
\end{proof}
%%
%% AAPP
%%
%% 202412250907
%%
On the other hand, since Theorem~3 in the paper
\cite{DBLP:journals/siamdm/Epstein06} is proven using a sorted item
sequence, a lower bound for pre-sorted online bounded space algorithms
is straightforwardly derived.
\begin{theorem}[\cite{DBLP:journals/siamdm/Epstein06}]
  For any $k \geq 2$, any pre-sorted online $B$-bounded space $k$-cardinality
  constrained algorithm $ALG$, and any positive integer $m$, there
  exists a sorted item sequence $I$ such that $OPT_{k}(I) = m$ and
  $ALG(I) \geq \lambda_{k} \cdot OPT_{k}(I) - B k$.
  \label{thm:k-item_online_bounded_lower}
\end{theorem}
\subsection{A Tight Bound for Online 1-Bounded Space Algorithms}
\label{subsec:k-item_online_general_unit}
It has long been known that $NF$ (Next Fit) is a best possible online
%% 1-bounded space algorithms for the classic bin packing
1-bounded space algorithm for the classic bin packing
problem~\cite{johnson1973near}.  Let $NF_{k}$ denote a $k$-cardinality
constrained version of algorithm $NF$ that opens a new one every time
it has packed $k$ items into a bin.  In this section, we prove that
$NF_{k}$ is best possible among online 1-bounded space $k$-cardinality
constrained algorithms as well, achieving $R_{NF_{k}} = 3 -
\frac{2}{k}$.
%% For lack of space, we omit the proofs of Theorems~\ref{thm:k-item_online_general_unit_nextfit_k}
%% and~\ref{thm:k-item_online_general_unit_lower}.
%% See
%% Appendices~\ref{subsec:thm:k-item_online_general_unit_nextfit_k}
%% and~\ref{subsec:thm:k-item_online_general_unit_lower} for the proofs
%% of Theorems~\ref{thm:k-item_online_general_unit_nextfit_k}
%% and~\ref{thm:k-item_online_general_unit_lower}, respectively.
%%
%% 202410301104
%%
\begin{theorem}
  For any $k \geq 2$, it holds that $NF_{k}(I) \leq \left(3 -
  \frac{2}{k}\right) \cdot OPT_{k}(I) + 1$
  for all item sequences $I$.
  \label{thm:k-item_online_general_unit_nextfit_k}
\end{theorem}
\begin{proof}
  We use the following weight function:
  \begin{alignat*}{2}
    w_{4}(x) & =
    \begin{cases}
      2 - \frac{1}{k}, & x \in (1 - \frac{1}{2k}, 1]; \\
        2 x, & x \in (\frac{1}{2k}, 1 - \frac{1}{2k}];\\
          \frac{1}{k}, & x \in (0, \frac{1}{2k}].
    \end{cases}
  \end{alignat*}
  Given $I = (a_{1}, a_{2}, \ldots, a_{n})$, define $w_{4}(I) =
  \sum_{i=1}^{n} w_{4}(a_{i})$.
  Let $I = (a_{1}, a_{2}, \ldots, a_{n}) \in (0, 1]^{n}$ be an
    arbitrary item sequence.

    (I)~We are going to prove that $w_{4}(I) \geq NF_{k}(I) -
    1$. Consider the assignment by $NF_{k}$ for $I$ as a sequence of
    bins. The assignment consists of two types of bins: bins that
    contain exactly $k$ items which we call \emph{full} bins, and bins
    that contain less than $k-1$ items which we call \emph{non-full}
    bins.  The case where the assignment by $NF_{k}$ for $I$ contains
    only non-full bins is discussed at the end of (I). Below, we first
    consider the case where the assignment contains at least one full
    bin.

    Remove the latest bin and cut the remaining sequence of
    $(NF_{k}(I) - 1)$ bins after each full bin. Then, we get a
    concatenation of subsequences with zero or more consecutive
    non-full bins, followed by one full bin. To prove that $w_{4}(I)
    \geq NF_{k}(I) - 1$, it is sufficient to show that the sum of the
    weights of the items in the bins in each subsequence is greater
    than or equal to the number of bins in the subsequence.
    See Figure~\ref{fig:k-item_online_general_unit_nextfit_k}.
  \begin{figure}[t]
    \centering
    \includegraphics[width=1.0\linewidth]
                    {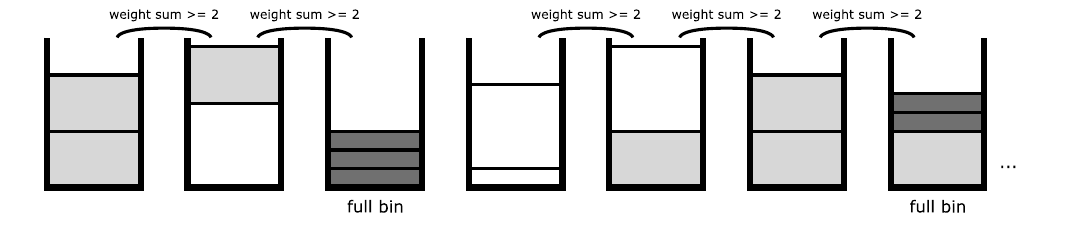}
    \caption{The assignment by $NF_{k}$ for the case of $k = 3$.
    The bins without ``full bin'' labels are non-full bins.
    }
    \label{fig:k-item_online_general_unit_nextfit_k}
  \end{figure}

  (i)~The case where the subsequence has an even number of non-full
  bins. We first mention that the sum of the weights of the items in
  the last full bin is $k \cdot \frac{1}{k} \geq 1$, since any small
  item has the lowest weight $\frac{1}{k}$.

  Look at consecutive non-full bin pairs in the subsequence, say, of
  indices $l$ and $l+1$.  The sum of the sizes of the items in these
  two bins must exceed 1 in total, since otherwise $NF_{k}$ would have
  packed some item in bin $l+1$ into the bin $l$.

  Below we show that the sum of the weights of the items in bins $l$
  and $l+1$ exceeds 2.  If the sum of the sizes of the items in each
  of bins $l$ and $l+1$ exceeds $\frac{1}{2}$, the proof is complete
  because the sum of the weights of the items in each bin exceeds $2
  \cdot \frac{1}{2} = 1$.  In the following, we assume that the sum of
  the sizes of the items in bin $l$ is less than or equal to
  $\frac{1}{2}$. (The case where the sum of the sizes of the items in
  bin $l+1$ is less than or equal to $\frac{1}{2}$ can be similarly
  shown).
  (i-a)~When the sum of the sizes of the items bin $l$ is less than or
  equal to $\frac{1}{2k}$. Since any small item has a minimum weight
  of $\frac{1}{k}$, the sum of the weights of the items in bin $l$ is
  at least $\frac{1}{k}$. On the other hand, the sum of the sizes of
  the items of bin $l+1$ must be greater than $1 - \frac{1}{2k}$, and
  the sum of the weights in the items there is greater than $2 -
  \frac{1}{k}$. That is, the weight sum exceeds $\frac{1}{k} + (2 -
  \frac{1}{k}) = 2$.
  (i-b)~When the sum of the sizes of the items of bin $l$ is $a$ such
  that $\frac{1}{2k} < a \leq \frac{1}{2}$, the sum of the weights is
  greater than $2 a$. On the other hand, the sum of the sizes of the
  items of bin $l+1$ should be greater than $1 - a$.  Since $1 - a < 1
  - \frac{1}{2k}$, the sum of the weights in the items there is
  greater than $2(1 - a)$. Therefore, the weight sum is $2a + 2(1 - a)
  > 2$.

  (ii)~The case where the subsequence has an odd number of non-full
  bins. Starting from the earliest bin, we sequentially look at pairs
  of non-full bins in the subsequence. As discussed in (i), the weight
  sum of each pair turns to be greater than 2. The rest is a pair of a
  non-full bin and a full bin which is the latest. The sum of the
  weights of the items in each bin must exceed 1, therefore the weight
  sum of this pair is greater than 2 and as revealed in~(i),

  If there are only non-full bins in the assignment by $NF_{k}$ for
  $I$, as discussed in~(i), the sum of the weights of every pair,
  except the latest bin, is greater than or equal to 2. Thus,
  $w_{4}(I) \geq NF_{k}(I) - 1$ holds.

  (II)~Next, we prove $w_{4}(I) \leq \left(3 - \frac{2}{k}\right)
  \cdot OPT_{k}(I)$. To do so, it is sufficient to show that for any
  item sequence with at most $k$ items that can fit in a bin,
  the sum of the weights is at most
  $3 - \frac{2}{k}$.

  This sum of the weights $3 - \frac{2}{k}$ can be achieved with a $(1
  - \varepsilon)$-item and $(k - 1)$ of $\varepsilon$-items such that
  $0 < \varepsilon \leq \frac{2}{k}$. In fact, we have $2 -
  \frac{1}{k} + (k - 1) \cdot \frac{1}{k} = 3 - \frac{2}{k}$.

  Every item sequence that maximizes the sum of the weights consists
  of exactly $k$ items, since any item sequence less than $k$ can
  increase its weight by replacing each item with a slightly smaller
  one without changing the weight, and
  adding an infinitesimal item of weight $\frac{1}{k}$.

  Consider a sequence of exactly $k$ items that can be packed into a
  bin and includes at least two items of size over $\frac{1}{2k}$.  The sum of
  the weights of the items of size over $\frac{1}{2k}$ is less than $2$,
  since the sum of their sizes is less than 1.  Thus, the total weight
  is less than $2 + (k - 2) \cdot \frac{1}{k} = 3 - \frac{2}{k}$.
  This implies that an item sequence that maximizes the sum of the
  weights must have at most one item of size over $\frac{1}{2k}$, such
  as the item sequence mentioned above.

  The theorem is shown by~(I) and~(II).
\end{proof}
%%
%% AAPP
%%
%% 202410301104
%%
\begin{theorem}
  For any online 1-bounded space $k$-cardinality constrained algorithm
  and any positive integer $m \geq 3$, there exists an item sequence
  $I$ such that $OPT_{k}(I) = m$ and: (A)~$ALG(I) \geq 2 \cdot
  OPT_{k}(I) - 2$ for $k = 2$ and (B)~$ALG(I) \geq \left(3 -
  \frac{2}{k}\right) \cdot OPT_{k}(I) - 5 + \frac{5}{k}$ for each $k
  \geq 3$.
  \label{thm:k-item_online_general_unit_lower}
\end{theorem}
\begin{proof}
  Fix an online 1-bounded space algorithm $ALG$ and a positive integer
  $m$.  Set $N = m - 1$.  Choose $\varepsilon$ and $\delta$ such that
  $0 < N \varepsilon < \frac{1}{2}$ and $0 < (k - 2) \cdot \delta <
  \varepsilon$.  Set $a_{i} = \frac{1}{2} + i \cdot \varepsilon - (k -
  2) \cdot \delta$ and $b_{i} = \frac{1}{2} - (i-1) \cdot \varepsilon$
  for $i = 1, 2, \ldots, N$.  Take
  \begin{equation*}
    I = (a_{1}, b_{1}, a_{2}, b_{2}, \ldots, a_{N}, b_{N}, 
    \overbrace{\delta, \ldots, \delta}^{(k - 2)(N - 1)}
     ).
  \end{equation*}
  Note that $I$ has no $\delta$-items for $k = 2$.

  It is seen that $OPT_{k}(I) = N + 1 (= m)$ holds. In fact, for $i =
  1, 2, \ldots, N-1$, the sum of the sizes of the $a_{i}$-item, the
  $b_{i+1}$-item, and $(k - 2)$ of $\delta$-items is exactly 1. They
  fill $(N-1)$ bins to capacity.  The remaining items are the
  $b_{1}$-item and the $a_{N}$-item.  which can only be packed one per
  bin. Thus, the total number of bins is $(N - 1) + 2 = N - 1$. See
  Figure~\ref{fig:k-item_online_general_unit_lower}.
  \begin{figure}[t]
    \centering
    \includegraphics[width=1.0\linewidth]
                    {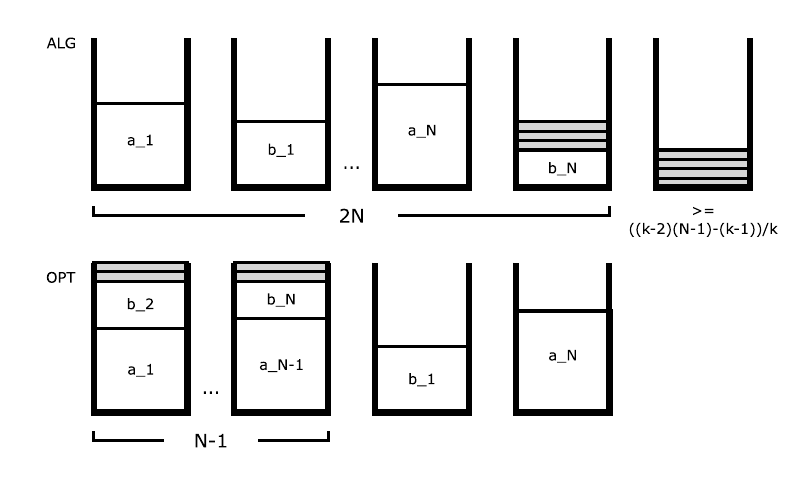}
    \caption{
      The assignments by $ALG$ and $OPT_{k}$ for $I$ for the case of $k = 4$.
      The light grey items are $\delta$-items.
    }
    \label{fig:k-item_online_general_unit_lower}
  \end{figure}

  On the other hand, let us evaluate $ALG(I)$.  Note that $ALG$ is a
  1-bounded space algorithm and therefore $ALG$ packs each item always
  into the latest bin, not into any earlier bin.  Since $a_{i} + b_{i}
  = 1 + \varepsilon - (k - 2) \cdot \delta > 1$ and $b_{i} + a_{i+1} =
  1 + 2 \varepsilon - (k - 2) \cdot \delta > 1$,
  $ALG$ must pack the $a_{1}$-item through the $b_{N}$-item one per bin.
  For $k = 2$, we immediately obtain $ALG(I) \geq 2N$.  For $k \geq
  3$, the best choice for the remaining $(k - 2)(N - 1)$ of
  $\delta$-items is to pack $(k - 1)$ of them into the bin containing
  the $b_{N}$-item and to pack the rest into new bins, which leads us to
  $ALG(I) \geq 2N + \frac{(k-2)(N-1) - (k-1)}{k} = \left(3 -
  \frac{2}{k}\right) N - 2 + \frac{3}{k}$.

  Hence, for $k = 2$, we have $ALG(I) \geq 2 \cdot OPT_{k}(I) - 2$,
  which is the claim~(A).  For $k \geq 3$, we get $ALG(I) \geq \left(3
  - \frac{2}{k}\right) (OPT_{k}(I) - 1) - 2 + \frac{3}{k} = \left(3 -
  \frac{2}{k}\right) \cdot OPT_{k}(I) - 5 + \frac{5}{k}$, which is the
  claim~(B).
\end{proof}
%% 
%% AAPP
%%
\section{Future Work}
\label{sec:future_work}
It is interesting to narrow the bounds for each subclass of
algorithms.  Are $MM$ and $MM_{k}$ best possible max-min 1-bounded
space algorithms?  Our conjecture is ``yes''.  Is there a better
pre-sorted online unbounded space algorithm than $FFD$?  Furthermore,
recall that Theorem~\ref{thm:none_online_sorted_bounded_lower} claims
that for the classic bin packing problem, the tight bounds for
pre-sorted online 1-bounded space algorithms and pre-sorted online
bounded space algorithms coincide.  Then, is there a gap between
max-min 1-bounded space algorithms and max-min bounded space
algorithms?
%%
%% \begin{credits}
%% \subsubsection*{Acknowledgments.}
%% \section*{Acknowledgments.}
%% \subsubsection{\ackname}
%%
%% \section*{Acknowledgments}
%%   We are grateful to Toshiki Tsuchiya for useful discussions.  We
%%   would like to thank Hiromu Hashimoto for pointing out an error in
%%   the manuscript.
%%   This work was supported by JSPS KAKENHI Grant Numbers
%%   25K14990,
%%   20K11689,
%%   20K11676,
%%   23K11100, and
%%   20K11808.
%% %%
%% \end{credits}
%%
\bibliographystyle{alpha}% bib style
\bibliography{main}% your bib database

\newcommand{\etalchar}[1]{$^{#1}$}
\begin{thebibliography}{BCKK04}

\bibitem[BBD{\etalchar{+}}18]{conf/esa/BaloghBDEL18}
J.~Balogh, J.~B{\'{e}}k{\'{e}}si, Gy. D{\'{o}}sa, L.~Epstein, and A.~Levin.
\newblock A new and improved algorithm for online bin packing.
\newblock In {\em Proc. ESA '18}, pages 5:1--5:14, 2018.

\bibitem[BBD{\etalchar{+}}20]{BALOGH202034}
J.~Balogh, J.~B\'{e}k\'{e}si, Gy. D\'{o}sa, L.~Epstein, and A.~Levin.
\newblock Online bin packing with cardinality constraints resolved.
\newblock {\em J. Comput. Syst. Sci.}, 112:34--49, 2020.

\bibitem[BBD{\etalchar{+}}21]{10.1007/s00453-021-00818-7}
J.~Balogh, J.~B\'{e}k\'{e}si, Gy. D\'{o}sa, L.~Epstein, and A.~Levin.
\newblock A new lower bound for classic online bin packing.
\newblock {\em Algorithmica}, 83(7):2047–2062, July 2021.

\bibitem[BBG12]{DBLP:journals/tcs/BaloghBG12}
J.~Balogh, J.~B{\'e}k{\'e}si, and G.~Galambos.
\newblock New lower bounds for certain classes of bin packing algorithms.
\newblock {\em Theor. Comput. Sci.}, 440-441:1--13, 2012.

\bibitem[BC81]{10.1137/0602019}
B.~S. Baker and E.~G. Coffman, Jr.
\newblock A tight asymptotic bound for {Next-Fit-Decreasing} bin-packing.
\newblock {\em {SIAM} J. Algebraic Discrete Methods}, 2(2):147–152, 1981.

\bibitem[BCKK04]{DBLP:journals/dam/BabelCKK04}
L.~Babel, B.~Chen, H.~Kellerer, and V.~Kotov.
\newblock Algorithms for on-line bin-packing problems with cardinality
  constraints.
\newblock {\em Discrete Appl. Math.}, 143(1-3):238--251, 2004.

\bibitem[BNR03]{DBLP:journals/algorithmica/BorodinNR03}
A.~Borodin, M.~N. Nielsen, and C.~Rackoff.
\newblock {(Incremental)} priority algorithms.
\newblock {\em Algorithmica}, 37(4):295--326, 2003.

\bibitem[CKPT17]{CHRISTENSEN201763}
H.~I. Christensen, A.~Khan, S.~Pokutta, and P.~Tetali.
\newblock Approximation and online algorithms for multidimensional bin packing:
  A survey.
\newblock {\em Comput. Sci. Rev.}, 24:63--79, 2017.

\bibitem[DLHT13]{DOSA201313}
Gy. D{\'o}sa, R.~Li, X.~Han, and Z.~Tuza.
\newblock Tight absolute bound for {First Fit Decreasing} bin-packing: {$FFD(L)
  \leq 11/9 OPT(L)+6/9$}.
\newblock {\em Theor. Comput. Sci.}, 510:13--61, 2013.

\bibitem[D{\'o}s07]{Dsa2007TheTB}
Gy. D{\'o}sa.
\newblock The tight bound of {First Fit Decreasing} bin-packing algorithm is
  {$FFD(I) \leq 11/9OPT(I) + 6/9$}.
\newblock In {\em Proc. ESCAPE '07}, pages 1--11, 2007.

\bibitem[Eps06]{DBLP:journals/siamdm/Epstein06}
L.~Epstein.
\newblock Online bin packing with cardinality constraints.
\newblock {\em {SIAM} J. Discrete Math.}, 20(4):1015--1030, 2006.

\bibitem[FW98]{DBLP:conf/dagstuhl/1996oa}
A.~Fiat and G.~J. Woeginger, editors.
\newblock {\em Online Algorithms, The State of the Art (the book grown out of a
  Dagstuhl Seminar, June 1996)}, volume 1442 of {\em LNCS}. Springer, 1998.

\bibitem[GJ79]{Garey:1979:CIG:578533}
M.~R. Garey and D.~S. Johnson.
\newblock {\em Computers and Intractability: A Guide to the Theory of
  NP-Completeness}.
\newblock W. H. Freeman \& Co., New York, NY, USA, 1979.

\bibitem[Joh73]{johnson1973near}
D.~S. Johnson.
\newblock {\em Near-optimal Bin Packing Algorithms}.
\newblock PhD thesis. Massachusetts Institute of Technology, 1973.

\bibitem[KV12]{Korte:2012:COT:2190621}
B.~Korte and J.~Vygen.
\newblock {\em Combinatorial Optimization: Theory and Algorithms}.
\newblock Springer Publishing Company, Incorporated, 5th edition, 2012.

\bibitem[LL85]{Lee:1985:SOB:3828.3833}
C.~C. Lee and D.~T. Lee.
\newblock A simple on-line bin-packing algorithm.
\newblock {\em J. ACM}, 32(3):562--572, 1985.

\bibitem[MES25]{MOMMESSIN2025106860}
C.~Mommessin, T.~Erlebach, and N.~V. Shakhlevich.
\newblock Classification and evaluation of the algorithms for vector bin
  packing.
\newblock {\em Comput. Oper. Res.}, 173:106860, 2025.

\bibitem[Tin95]{Tinhofer1995BinpackingAM}
G.~Tinhofer.
\newblock Bin-packing and matchings in threshold graphs.
\newblock {\em Discret. Appl. Math.}, 62:279--289, 1995.

\bibitem[YB08]{DBLP:journals/jco/YeB08}
Y.~Ye and A.~Borodin.
\newblock Priority algorithms for the subset-sum problem.
\newblock {\em J. Comb. Optim.}, 16(3):198--228, 2008.

\bibitem[Zhu16]{ZHU201683}
D.~Zhu.
\newblock Max-min bin packing algorithm and its application in nano-particles
  filling.
\newblock {\em Chaos Solit. Fractals}, 89:83--90, 2016.

\end{thebibliography}
\end{document}